 \documentclass[a4paper,twocolumn,accepted=2023-02-21]{quantumarticle}
 \pdfoutput=1
 
\usepackage[numbers,sort&compress]{natbib}

\usepackage[utf8]{inputenc}

\usepackage{graphicx}

\usepackage{amsmath,amssymb,array}
\usepackage{siunitx}
\usepackage{longtable,tabularx}

\usepackage{algorithm}
\usepackage{algorithmic}

\usepackage{multirow}
\usepackage{color}

\DeclareMathOperator{\Tr}{\mathrm{Tr}}

\newcommand{\set}[1]{\mathcal{#1}}
\newtheorem{theo}{Theorem}

\newtheorem{lemm}[theo]{Lemma}%
\newtheorem{rema}{Remark}%
\newtheorem{cor}[theo]{Corollary}%

\newcommand{\argmin}{\mathop{\rm argmin}\limits}

\def\R{{\mathbb R}}

\def\<{\langle}
\def\>{\rangle}

\def\QED{\mbox{\rule[0pt]{1.5ex}{1.5ex}}}

\newcommand{\qed}{\hfill \QED}

 \newenvironment{proofof}[1]{\vspace*{5mm} \par \noindent
{\it Proof of #1:\hspace{2mm}}}{\qed
}

\def\Label#1{\label{#1}\ [\ \text{#1}\ ]\ }
\def\Label{\label}

\begin{document}
\title{Efficient algorithms for quantum information bottleneck}

\author{Masahito Hayashi}
\email{hayashi@sustech.edu.cn}
\affiliation{Shenzhen Institute for Quantum Science and Engineering, Southern University of Science and Technology, Shenzhen,518055, China}
\affiliation{International Quantum Academy (SIQA), Futian District, Shenzhen 518048, China}
\affiliation{Guangdong Provincial Key Laboratory of Quantum Science and Engineering, Southern University of Science and Technology, Shenzhen, 518055, China}
\affiliation{Graduate School of Mathematics, Nagoya University, Nagoya, 464-8602, Japan}
\author{Yuxiang Yang}
\email{yuxiang@cs.hku.hk}
\affiliation{QICI Quantum Information and Computation Initiative, Department of Computer Science, The University of Hong Kong, Pokfulam Road, Hong Kong}

\begin{abstract}
The ability to extract relevant information is critical to learning. An ingenious approach as such is the information bottleneck, an optimisation problem whose solution corresponds to a faithful and memory-efficient representation of relevant information from a large system.
The advent of the age of quantum computing calls for efficient methods that work on information regarding quantum systems. Here we address this by proposing a new and general algorithm for the quantum generalisation of information bottleneck. Our algorithm excels in the speed and the definiteness of convergence compared with prior results. It also works for a much broader range of problems, including the quantum extension of deterministic information bottleneck, an important variant of the original information bottleneck problem.
Notably, we discover that a quantum system can achieve strictly better performance than a classical system of the same size regarding quantum information bottleneck, providing new vision on justifying the advantage of quantum machine learning.
\end{abstract}

\maketitle
\section{Introduction}
Learning is a task of eminent importance to the contemporary world. As such, it has always been of top priority to quest powerful tools for learning information.
Information bottleneck \cite{tishby1999} stands as an excellent example, with many useful applications including 
deep learning  \cite{tishby2015deep,shwartz2017opening,goldfeld2020information}, 
video processing \cite{10.1145/1180639.1180654},
clustering \cite{10.1145/345508.345578} and polar coding \cite{8368978}.
Concretely, information bottleneck is a method to 
extract a piece of information $T$
with respect to the system $Y$ from the system $X$, and
is formulated as the minimization problem of the difference $I(T:X)-\beta I(T:Y)$ with a positive parameter $\beta$, where $I(T:X)$ is the mutual information between $T$ and $X$.
In particular, we are interested in the case when $X$ is classical.
By design, information bottleneck achieves an irreversible compression, by extracting essential information about $Y$ and simultaneously removing unessential information contained in $X$.

\begin{figure}[tb]{}
\begin{center} 
\includegraphics[width=\linewidth]{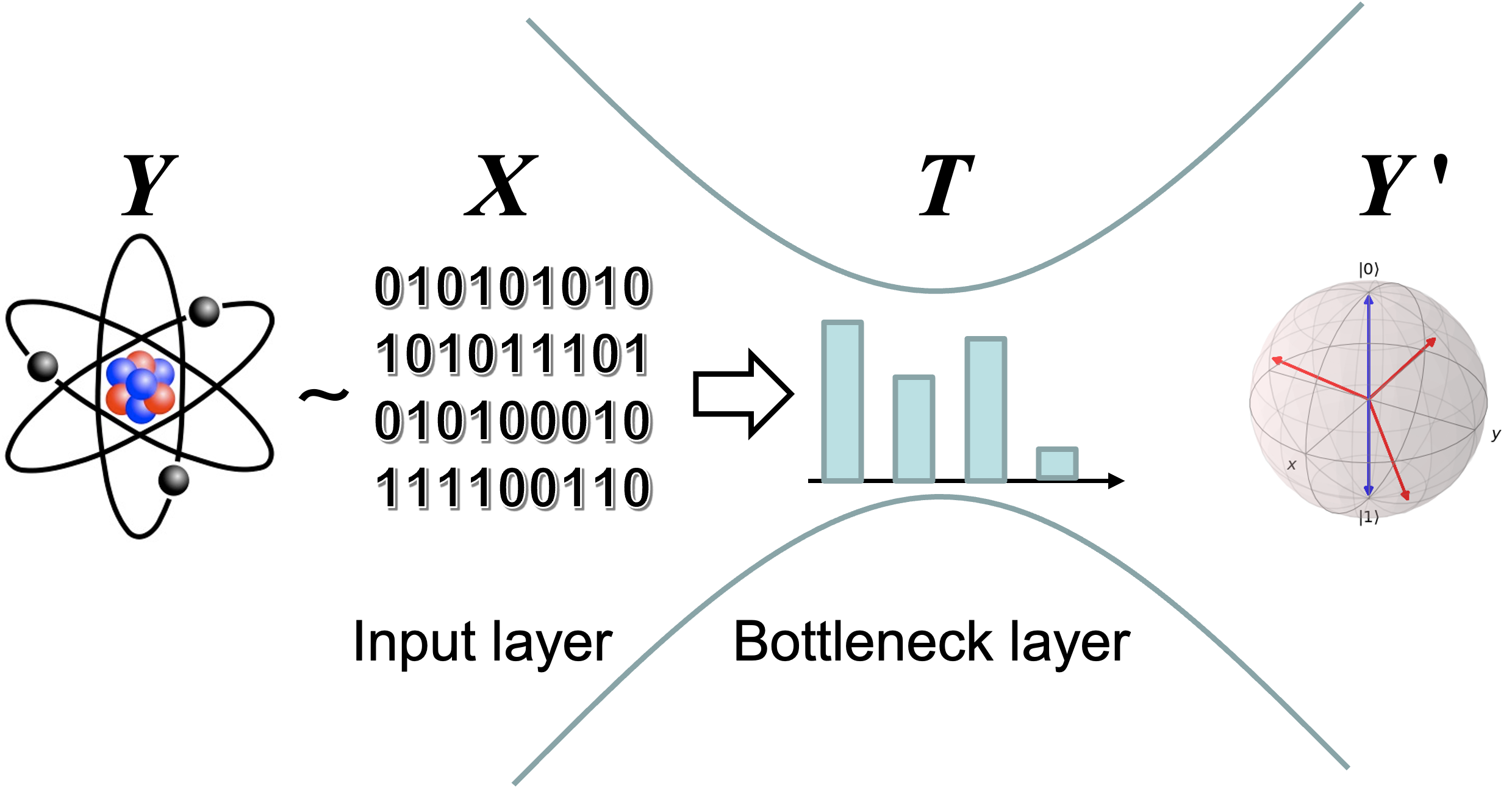}
\end{center}
\caption{{\bf Visualization of quantum information bottleneck.}
In a prototypical setting of quantum information bottleneck, the task is to compress a classical system into a smaller system $T$, which can be either classical or quantum, by extracting its useful information about a quantum system $Y$ and removing the useless information. It is expected that more relevant information $Y'$ about $Y$, instead of the entire $X$, can be recovered from $T$.
}\Label{F1} 
\end{figure}

As we are stepping into the age of quantum information, 
the demand is growing for a method that efficiently learns information on a quantum system. For this purpose, let us consider the setup of quantum information bottleneck (QIB), demonstrated in Fig.~\ref{F1}.
Similar as its classical counterpart, the aim of QIB is to compress $X$ into a smaller system $T$ while preserving the correlation with $Y$ when some of these systems are quantum systems. 
Prior to this work, QIB has been discussed in several recent works 
\cite{PhysRevA.94.012338,8513885,8849518,9174416,PRXQuantum.2.040321} and has been applied to quantum information theory \cite{8849518,9174416}  and quantum machine learning \cite{PRXQuantum.2.040321}.
On the other hand, the fundamental properties of QIB such as convergence have not been analysed, which hinders its application in more practical tasks. 
Quantum information bottleneck is first proposed as a quantum extension of information bottleneck method in \cite{PhysRevA.94.012338}.
It also derived a necessary condition
for the solution of the minimization problem (see \cite[Appendix A]{PhysRevA.94.012338})
by using Lagrange multiplier method in the same way as \cite{1054753,1054855}.
Using the obtained condition, it also proposed an iterative algorithm to find a solution to 
satisfy the necessary condition
\cite[Appendix C]{PhysRevA.94.012338}. 
Then, the reference \cite{8513885} considered QIB in the quantum communication scenario. 
\footnote{The reference \cite[Appendix A]{8513885} 
derived a necessary condition
for the solution of the minimization problem 
by using Lagrange multiplier method in the same way as \cite{1054753,1054855}.
Using the obtained condition, it also proposed an iterative algorithm to find a solution to 
satisfy the necessary condition
\cite[The end of Appendix C]{8513885}.}
However, no study discusses the behaviour of the iterative algorithm, i.e., 
it is not known whether the algorithm monotonically reduces the objective function
\cite{tishby1999,10.1162/NECO_a_00961,PhysRevA.94.012338,8513885}.
It was also claimed in \cite[Appendix B]{PhysRevA.94.012338} that there is no advantage of using a quantum $T$ if $X,Y$ are both classical.


In this work, we conduct a systematic study on quantum information bottleneck, focusing on the case when the system $X$ is classical. 
Compared to existing works \cite{PhysRevA.94.012338,8513885,8849518,9174416,PRXQuantum.2.040321}, our work makes significant contributions in several directions:

First, we provide throughout analyses on two critical properties -- efficiency and convergence -- of QIB. 
Motivated by a recent generalization \cite{Ramakrishnan} of the Arimoto-Blahut algorithm \cite{1054753,1054855}, we introduce a new quantum information bottleneck algorithm with an acceleration parameter $\gamma$ that can make the value of QIB converges much faster than before when chosen properly. 
We prove rigorous criteria for our algorithm to converge and to achieve a minimum. In particular, we prove that the choice of $\beta$ plays an important role in convergence.

Second, in contrast to the claim in Refs.~\cite{PhysRevA.94.012338,8513885}, we provide concrete examples where using a quantum instead of classical $T$ could reduce the minimal value of QIB.
Notably, our result justifies a genuine quantum advantage in quantum machine learning \cite{wittek2014quantum,schuld2015introduction,biamonte2017quantum}, where the employment of quantum circuits has been prevalent \cite{schuld2019quantum,havlivcek2019supervised,blank2020quantum,lloyd2020quantum,perez2020data,schuld2021supervised} but the quantum advantage was rarely justified.

Last but not least, we generalise QIB by considering a general
target function $(1-\alpha)H(T)+\alpha I(T:X)-\beta I(T:Y)$ with parameters $\alpha,\beta\ge0$, which reduces to the standard QIB when $\alpha=1$. By doing so, the generalised QIB contains QDIB, i.e., the quantum version of deterministic information bottleneck \cite{10.1162/NECO_a_00961}, by setting $\alpha=0$.
We show that our analyses and our algorithm hold for this generalised setting and, in particular, to QDIB.
Then, we clarify that QDIB can be used to find a good approximate sufficient statistics $T$ for 
$X$ for $Y$, which requires a smaller entropy $H(T)$ and larger mutual information
$I(T:Y)$.
We justify our finding via a numerical example, where QDIB extracts a good approximate sufficient statistics over information about a quantum ensemble.

In summary, our work addresses several critical issues of QIB, including convergence, efficiency, choice of parameters, and the quantum advantage. We also extend QIB to a generalised setting and introduce the notion of QDIB. Our results consist of both rigorous analytical analyses and numerical experiments that justifies the importance of QIB and QDIB in fundamental tasks of learning.


The remaining part of this paper is organized as follows.
Section \ref{sec-QIB} introduces our algorithm for  
quantum information bottleneck, and discusses its convergence and dependence of the 
parameter $\beta$.
Section \ref{SIV} discusses our algorithm when our memory system $T$ is classical.
Section \ref{S4} presents examples that  realizes
a smaller value of the target function by quantum memory $T$  
than by classical memory $T$.
Section \ref{S8} discusses an application of our QIB algorithm in data classification.
Section \ref{S6} proposes our algorithm for 
quantum deterministic information bottleneck, and studies its properties.
Section \ref{SQ} applies it to the extraction of approximate sufficient statistics,
and numerically verifies its efficiency in an example.
Section \ref{S9} makes discussion and conclusion.

\section{The quantum information bottleneck (QIB) problem}\Label{sec-QIB}

\subsection{Problem definition}
Consider a classical-quantum joint system composed of $X$ and $Y$ with the joint state 
\begin{align}\Label{joint-state}
\rho_{XY}:= \sum_{x}P_X(x)|x\rangle \langle x|\otimes \rho_{Y|x},
\end{align}
where $X$ is a classical system and $Y $ is a quantum system. 
Our quantum information bottleneck (QIB) problem aims at constructing an information processor, modelled by a c-q channel $\sigma_{T|X}$ from $X$ to $T$ (which prepares a quantum state $\sigma_{T|x}$ when the classical register is $x$), that extracts efficient information from $X$ with respect to the quantum system $Y$. After the action of the information processor, the joint state becomes:
\begin{align}
\rho_{XYT}:= \sum_{x}P_X(x)|x\rangle \langle x|\otimes \rho_{Y|x}\otimes \sigma_{T|x}.
\end{align}
To this aim, the QIB problem concerns constructing a classical-quantum channel $\sigma_{T|X}:X\to T$ that minimizes the information bottleneck function, consisting of entropic quantities defined with respect to the joint state $\rho_{XYT}$:
\begin{align}\Label{IB-function}
f_{\alpha}(\sigma_{T|X})&:=H(T)-\alpha H(T|X)-\beta I(T:Y) \nonumber \\
&= (1-\alpha)H(T)+\alpha I(T:X)-\beta I(T:Y),
\end{align}
where $H(T)$ denotes the entropy of $T$ \footnote{For convenience, the notation $H(A)$ stands for the Shannon entropy when the system $A$ is classical and for the von Neumann entropy when $A$ is quantum.}, $H(T|X)$ denotes the conditional entropy of $T$ on $X$, and $I(T:Y)$ stands for the mutual information between $T$ and $Y$. 

That is, our aim is the calculation of the following value:
\begin{align}\label{QIB-problem}
{\cal I}_{\alpha,\beta}:= \min_{\sigma_{T|X}}f_{\alpha}(\sigma_{T|X}).
\end{align}

In the information bottleneck (\ref{IB-function}), $\alpha$ and $\beta$ are positive real variables modelling the objective of the task. In the original proposal of information bottleneck \cite{tishby1999} $\alpha=1$. Another common choice of $\alpha$ is $\alpha=0$, and the task is called a deterministic QIB (whose classical counterpart was discussed in Ref.~\cite{10.1162/NECO_a_00961}). The parameter $\beta$ controls the tradeoff between faithfulness and compression. For instance, in a deterministic information bottleneck, a larger $\beta$ would make $I(T:Y)$ more prominent in the objective function, forcing the information processor to preserve more information about $Y$, whereas a smaller $\beta$ would signify the role of $I(T:X)$, prompting the information processor to do more compression in $X$.

Although 
this section addresses the case with 
quantum systems $Y$ and $T$,
the case with a classical system $Y$ and a quantum $T$ 
can be contained as a special case by considering the diagonal densities $\rho_{Y|x}$.
On the other hand, the case with a classical system $T$ 
is a different problem from 
the case with a quantum system $T$
because we need to discuss a different minimization problem, which has a different range for the minimizing variable.
Fortunately, our algorithm for a quantum system $T$, presented in the next subsection, can be applied to the case with a classical system $T$.
Section \ref{SIV} discusses the case of $T$ being classical.
We remark that the case where both $T$ and $Y$ are classical has been widely studied in classical information theory and machine learning; see, e.g., Refs.~\cite{tishby1999,tishby2015deep,10.1162/NECO_a_00961,shwartz2017opening}.

\subsection{QIB algorithm for $\alpha=1$}
The paper \cite{PhysRevA.94.012338} discussed this problem when 
$X,Y,T$ are quantum systems and $\alpha=1$,
extending the classical information bottleneck \cite{tishby1999} to the quantum regime.
It derived a necessary condition for $\sigma_{X|T}$ to achieve the minimum (\ref{QIB-problem}). 
The necessary condition 
with quantum systems $T,Y$ and a classical system $X$ 
is written as 
\begin{align}
\log \sigma_{T|x}=
& (1-\beta) \log \sigma_T[\sigma_{T|X}]  \nonumber \\
&-\beta \Tr_Y \rho_{Y|x} \Big( 
\log  
\rho_Y 
-
\log \sigma_{YT} [\sigma_{T|X}]
\Big)-C_x,\Label{MM6}
\end{align} 
 where $C_x$ is a normalizing constant and
\begin{align}
\rho_Y:=&\sum_{x} P_X(x)\rho_{Y|x}\\
\sigma_T[\sigma_{T|X}]:=&\sum_x P_X(x) \sigma_{T|x}\\
\sigma_{YT}[\sigma_{T|X}]:=&\sum_x P_X(x) \sigma_{T|x}\otimes \rho_{Y|x}.
\end{align}

Since this condition is self-consistent, 
using this condition, the paper \cite{PhysRevA.94.012338} proposed the following 
iterative algorithm with the following update rule:
\begin{align}
\sigma_{T|x}^{(n+1)}:=&
\frac{1}{e^{C_x}}\exp \Big( (1-\beta) \log  \sigma_T[\sigma_{T|X}^{(n)}]  \nonumber \\
&-\beta \Tr_Y \rho_{Y|x} \Big( 
\log  
\rho_Y 
-
\log \sigma_{YT} [\sigma_{T|X}^{(n)}]
 \Big) 
\Big).\Label{ACY}
\end{align} 

\subsection{The acceleration parameter $\gamma$}
Next, we propose an extension of the iterative algorithm in \cite{PhysRevA.94.012338}. 
First, we introduce a new parameter $\gamma>0$ and rewrite the condition \eqref{MM6} as:
\begin{align}
&\log \sigma_{T|x}=
(1-\frac{1}{\gamma}) \log \sigma_{T|x}
+\frac{1}{\gamma} \log \sigma_{T|x} \nonumber \\
=&
(1-\frac{1}{\gamma}) \log \sigma_{T|x}
+\frac{1}{\gamma} (1-\beta) \log \sigma_T[\sigma_{T|X}]  \nonumber \\
&-\frac{1}{\gamma} \beta \Tr_Y \rho_{Y|x} \Big( 
\log  
\rho_Y 
-
\log \sigma_{YT} [\sigma_{T|X}]
\Big)-\frac{1}{\gamma} C_x  \nonumber \\
=&
 \log \sigma_{T|x}
- \frac{1}{\gamma} {\cal F}_1[\sigma_{T|X}](x)
-\frac{1}{\gamma} C_x ,\Label{ACAC}
\end{align} 
where
\begin{align}
& {\cal F}_1[\sigma_{T|X}](x)  \nonumber \\
:=&  
-\log \sigma_{T}[\sigma_{T|X}]+ \log \sigma_{T|x}  \nonumber \\
&+\beta \Tr_Y \Big( \rho_{Y|x}\Big(
\log  (\sigma_T[\sigma_{T|X}] \otimes \rho_Y )-
\log \sigma_{YT} [\sigma_{T|X}]
\Big) \Big) .
\end{align} 
Using \eqref{ACAC}, we can derive another iterative algorithm as
\begin{align}
\sigma_{T|x}^{(n+1)}:=
\frac{1}{e^{\frac{1}{\gamma} C_x}}\exp \Big( 
\log \sigma_{T|x}^{(n)}-\frac{1}{\gamma} {\cal F}_1[\sigma_{T|X}^{(n)}](x)\Big).\Label{ACM}
\end{align} 

In this way, we can easily generalize the iterative algorithm \eqref{ACY} by \cite{PhysRevA.94.012338}.
However, it is not {trivial} to find the suitable value for $\frac{1}{\gamma}$, which, as we show later, is critical to the efficiency of our iterative algorithm.
Although many papers \cite{tishby1999,10.1162/NECO_a_00961,PhysRevA.94.012338,8513885} discussed 
the iterative algorithm given by \eqref{ACY} including the classical case,
no preceding study showed the convergence of the iterative algorithm by \eqref{ACY}.
In addition, the discussion above {focuses} on the case of $\alpha=1$ and does not include the case of deterministic information bottleneck ($\alpha=0$).
Therefore, to make an efficient algorithm, we need to discuss the choice of the parameter $\gamma$ for generic $\alpha$.

\subsection{QIB algorithm with general $\alpha$ and convergence}\Label{S2-C}
To analyze the convergence of the algorithm \eqref{ACM}, 
we introduce a two-input variable function
based on the idea in Ref.~\cite[Section III-B]{Ramakrishnan},
{whereas the method in Ref.~\cite[Section III-B]{Ramakrishnan}
was obtained as a generalization of the Arimoto-Blahut algorithm \cite{1054753,1054855}.}
The idea is that, instead of directly solving the minimization of $f_{\alpha}(\sigma_{T|X})$, which is often too difficult, we find a continuous function
$J(\sigma_{T|X},\sigma_{T|X}')$ with two variables $\sigma_{T|X},\sigma_{T|X}'$. Then we can update these two input variables $\sigma_{T|X},\sigma_{T|X}'$ alternately to decrease $J(\sigma_{T|X},\sigma'_{T|X})$. Finally, if the function satisfies
\begin{align}
f_{\alpha}(\sigma_{T|X})  
=
J(\sigma_{T|X},\sigma_{T|X})
\Label{BCRR},
\end{align}
the minimum of $J(\sigma_{T|X},\sigma'_{T|X})$ will be close to the minimum of the IB function.

The above type of functions can be constructed 
if we find {an} operator 
${\cal F}_\alpha[\sigma_{T|X}](x)$ to satisfy 
\begin{align}
f_{\alpha}(\sigma_{T|X})  
=&\sum_{x}  P_X(x) \Tr_T \sigma_{T|x} {\cal F}_\alpha[\sigma_{T|X}](x)\Label{BCOV},
\end{align}
In this paper, we employ the following function:
\begin{align}\Label{F-function}
& {\cal F}_\alpha[\sigma_{T|X}](x)  \nonumber \\
:=&  
-\log \sigma_{T}[\sigma_{T|X}]+\alpha\log \sigma_{T|x}  \nonumber \\
&+\beta \Tr_Y \Big( \rho_{Y|x}\Big(
\log  (\sigma_T[\sigma_{T|X}] \otimes \rho_Y )-
\log \sigma_{YT} [\sigma_{T|X}]
\Big) \Big) .
\end{align}
Then, the condition \eqref{BCOV} is satisfied.

Using this function, we can define 
$J_0(\sigma_{T|X},\sigma_{T|X}'):=
\Tr_T \sum_{x} \sigma_{T|x} P_X(x) {\cal F}_\alpha[\sigma_{T|X}'](x)$, which  satisfies the condition \eqref{BCRR}.
However, 
it is difficult to optimize two input variables alternately
in the function $J_0(\sigma_{T|X},\sigma_{T|X}')$.
Instead, for $\gamma >0$,
we introduce the following function
\begin{align}
& J_{\gamma,\alpha}( \sigma_{T|X},\sigma_{T|X}') \\
:= &
\gamma D( \sigma_{T|X}\| \sigma_{T|X}')
+\sum_{x} P_X(x) \Tr_T \sigma_{T|x} {\cal F}_\alpha[\sigma_{T|X}'](x),\Label{CPTV}
\end{align}
where 
$D( \sigma_{T|X}\| \sigma_{T|X}'):=\sum_{x}P_X(x)D( \sigma_{T|x}\| \sigma_{T|x}') 
$ and
$D( \sigma_{T|x}\| \sigma_{T|x}')$ denotes the relative entropy.

Next, we need to specify the rules of the alternatively updating $\sigma_{T|X},\sigma_{T|X}'$. Crucially, we need to ensure that $J_{\gamma,\alpha}(\sigma_{T|X},\sigma_{T|X}')$ is non-increasing under the updating rules.
To this purpose, we first introduce the following condition:
\begin{description}
\item[(A1)]
$\sigma_{T|X}$ and $\sigma_{T|X}'$ satisfy
the relation
\begin{align}
& \gamma \sum_{x}P_X(x)D( \sigma_{T|x}\|\sigma_{T|x}')
 \nonumber \\
\ge & \sum_{x} P_X(x) \Tr_T \sigma_{T|x} 
({\cal F}_\alpha[\sigma_{T|X}](x) -{\cal F}_\alpha[\sigma_{T|X}'](x))\Label{XLOV} .
\end{align}
\end{description}

In fact, the condition (A1) is rewritten as $
\gamma \ge \gamma(\sigma_{T|X},\sigma_{T|X}')$
by defining $\gamma(\sigma_{T|X},\sigma_{T|X}')$ as
\begin{align}
&\gamma(\sigma_{T|X},\sigma_{T|X}')  \nonumber \\
:=&
\frac{\sum_{x} P_X(x) \Tr_T \sigma_{T|x} 
({\cal F}_\alpha[\sigma_{T|X}](x) -{\cal F}_\alpha[\sigma_{T|X}'](x))}
{\sum_{x}P_X(x)D( \sigma_{T|x}\|\sigma_{T|x}')}.
\end{align}
This quantity is evaluated as
\begin{align}
\gamma(\sigma_{T|X},\sigma_{T|X}')
\le \alpha \Label{AMX}
\end{align}
because the relation
\begin{align}
& D (\rho_{YT} [\sigma_{T|X}]\|\rho_{YT} [\sigma_{T|X}'])
\ge D (\sigma_T[\sigma_{T|X}] \|
\sigma_T[\sigma_{T|X}'] ) \nonumber \\
=& D (\sigma_T[\sigma_{T|X}] \otimes \rho_Y\|
\sigma_T[\sigma_{T|X}'] \otimes \rho_Y)
\end{align}
implies the relation
\begin{align}
&\sum_{x}P_X(x)  \Tr_T \sigma_{T|x} 
({\cal F}_\alpha[\sigma_{T|X}](x) -{\cal F}_\alpha[\sigma_{T|X}'](x)) \nonumber \\
=& -D(\sigma_{T}[\sigma_{T|X}] \| \sigma_{T}[\sigma_{T|X}']) \nonumber \\
&+\alpha \sum_{x}P_X(x)D( \sigma_{T|x}\|\sigma_{T|x}') \nonumber \\
&+\beta D (\sigma_T[\sigma_{T|X}] \otimes \rho_Y\|
\sigma_T[\sigma_{T|X}'] \otimes \rho_Y)  \nonumber \\
&-\beta D (\rho_{YT} [\sigma_{T|X}]\|\rho_{YT} [\sigma_{T|X}'])
 \nonumber \\
\le &-D(\sigma_{T}[\sigma_{T|X}] \| \sigma_{T}[\sigma_{T|X}']) \nonumber \\
&+\alpha \sum_{x}P_X(x)D( \sigma_{T|x}\|\sigma_{T|x}') \nonumber \\
\le & \alpha \sum_{x}P_X(x)D( \sigma_{T|x}\|\sigma_{T|x}')
.\Label{BLZV}
\end{align}

To state our updating rules, we define 
\begin{align}
\hat{\sigma}_{\gamma,\alpha,T}[\sigma_{T|X}](x):=&
\exp\Big(\log \sigma_{T|x}-\frac{1}{\gamma}{\cal F}_\alpha[\sigma_{T|X}](x) \Big) \Label{BA1V}\\
\hat{\eta}_{\gamma,\alpha|x}[\sigma_{T|X}]
:=&\Tr \hat{\sigma}_{\gamma,\alpha,T}[\sigma_{T|X}](x) \Label{BAAV}\\
\hat{\sigma}_{\gamma,\alpha,T|x}[\sigma_{T|X}]:=
&\frac{1}{\hat{\eta}_{\gamma,\alpha}[\sigma_{T|X}](x)} 
\hat{\sigma}_{\gamma,\alpha,T}[\sigma_{T|X}](x).\Label{BA2V}
\end{align}
In particular, when $\gamma=\alpha$, 
the operator $\hat{\sigma}_{\gamma,\alpha,T}[\sigma_{T|X}](x)$
is simplified as
\begin{align}
&\hat{\sigma}_{\alpha,T}[\sigma_{T|X}](x)  \nonumber \\
=& \exp \Big(
\frac{1-\beta}{\alpha}\log \sigma_{T}[\sigma_{T|X}]
 \nonumber \\
&-\frac{\beta}{\alpha} \Tr_Y \rho_{Y|x} \Big( 
\log  
\rho_Y 
-
\log \sigma_{YT} [\sigma_{T|X}]
 \Big) 
\Big).
\end{align}

\begin{theo}\Label{LBCOV}
Under the condition (A1), we have
\begin{align}
J_{\gamma,\alpha}(\sigma_{T|X},\sigma_{T|X}') \ge & 
J_{\gamma,\alpha}(\sigma_{T|X},\sigma_{T|X})  \Label{BCOV2}\\
J_{\gamma,\alpha}( \sigma_{T|X},\sigma_{T|X}') \ge 
& J_{\gamma,\alpha}(
\hat{\sigma}_{\gamma,\alpha,T|X}[\sigma_{T|X}'],
\sigma_{T|X}').\Label{BCO2V}
\end{align}
\end{theo}

\begin{proofof}{Theorem \ref{LBCOV}}
The condition (A1) yields
\begin{align}
&J_{\gamma,\alpha}(\sigma_{T|X},\sigma_{T|X}) \nonumber \\
=
& \sum_{t}\Tr \sigma_{T|x}P_X(x) {\cal F}_\alpha[\sigma_{T|X}](x)  \nonumber \\
\le & \sum_{x}\Tr\sigma_{T|x}P_X(x) {\cal F}_\alpha[\sigma_{T|X}'](x,t) \nonumber  \\
&+ 
\gamma \sum_{x}P_X(x)D( \sigma_{T|x}\|\sigma_{T|x}')  \nonumber \\
=& J_{\gamma,\alpha}(\sigma_{T|X},\sigma_{T|X}').
\end{align}
Hence, we obtain \eqref{BCOV2}.

Also, we have
\begin{align}
&J_{\gamma,\alpha}( \sigma_{T|X},\sigma_{T|X}') \nonumber \\
\stackrel{(a)}{=}& \gamma \sum_{x}P_X(x)
\Tr \sigma_{T|x} \Big( \log \sigma_{T|x}- \log \sigma_{T|x}'  \nonumber \\
&+\frac{1}{\gamma} {\cal F}_\alpha[\sigma_{T|X}'](x) \Big) \nonumber \\
\stackrel{(b)}{=}& \gamma \sum_{x}P_X(x)
\Tr \sigma_{T|x} \Big( \log \sigma_{T|x}- 
{\hat{\sigma}_{\gamma,\alpha,T}[\sigma_{T|X}](x)}
\Big) \nonumber \\
\stackrel{(c)}{=}& \gamma \sum_{x}P_X(x)\Big(
\Tr \sigma_{T|x} \Big( \log \sigma_{T|x}- \log 
\hat{\sigma}_{\gamma,\alpha,T|x}[\sigma_{T|X}']\Big)  \nonumber \\
&-\log \hat{\eta}_{\gamma,\alpha}[\sigma_{T|X}'](x)\Big)  \nonumber \\
=& \gamma \sum_{x}P_X(x)
\big( D(\sigma_{T|x}\| 
\hat{\sigma}_{\gamma,\alpha,T|x}[\sigma_{T|X}']) \big)  \nonumber \\
&- 
\gamma \sum_{x}P_X(x)\log \hat{\eta}_{\gamma,\alpha|x}[\sigma_{T|X}']
\Label{ACO},
\end{align}
where $(a)$, $(b)$, and $(c)$ follow from 
\eqref{CPTV}, \eqref{BA1V}, and \eqref{BA2V},
 respectively. Finally, from Eq.~(\ref{ACO}) we can see that the minimum of $J_{\gamma,\alpha}( \sigma_{T|X},\sigma_{T|X}')$ is achieved when $\sigma_{T|X}=\hat{\sigma}_{\gamma,\alpha,T|x}[\sigma_{T|X}']$, since the first term of (\ref{ACO})  is non-negative (with equality achieved when $\sigma_{T|X}=\hat{\sigma}_{\gamma,\alpha,T|x}[\sigma_{T|X}']$) and the second term is independent of $\sigma_{T|X}$.
Hence,
 we obtain \eqref{BCO2V}. 
\end{proofof}

\begin{cor}\Label{CC8}
Assume that $\gamma \ge 
\sup_{\sigma_{T|X},\sigma_{T|X}'}\gamma(\sigma_{T|X},\sigma_{T|X}')$.
When $\sigma_{T|X}$ is a local minimizer, we have 
\begin{align}
 \hat{\sigma}_{\gamma,\alpha,T|x}[\sigma_{T|X}]=
\sigma_{T|X},\Label{AAV}
\end{align}
which is equivalent to \eqref{MM6} when $\alpha=1$.
\end{cor}

When $\gamma \ge \gamma(\hat{\sigma}_{\gamma,\alpha,T|X}[\sigma_{T|X}],\sigma_{T|X})$,
the following chain of inequalities hold: $
f_{\alpha}(\sigma_{T|X}) {=} J_{\gamma,\alpha}( \sigma_{T|X},\sigma_{T|X})  \ge 
J_{\gamma,\alpha}(\hat{\sigma}_{\gamma,\alpha,T|X}[\sigma_{T|X}],\sigma_{T|X})
\ge 
J_{\gamma,\alpha}(\hat{\sigma}_{\gamma,\alpha,T|X}[\sigma_{T|X}],
\hat{\sigma}_{\gamma,\alpha,T|X}[\sigma_{T|X}])=f_{\alpha}(\hat{\sigma}_{\gamma,\alpha,T|X}[\sigma_{T|X}])$.
Hence, the monotonicity of the information bottleneck under the updating rules is also guaranteed, as long as $\gamma$ is sufficiently large.
Finally, we propose the following algorithm with a fixed $\gamma$ and general $\alpha$:
\begin{algorithm}[H]
\caption{QIB algorithm}
\Label{protocol1V}
\begin{algorithmic}[1]
\STATE {\bf Input:} A joint state $\rho_{XY}$ [as in Eq.~(\ref{joint-state})]. 
\STATE {Randomly choose an initial c-q channel $\sigma_{T|X}^{(1)}$; }  
\STATE Create a counter $n$ as the number of iterations; initialize $n$ to 1. 
\REPEAT 
\STATE Choose $\sigma_{T|X}^{(n+1)}$ as $\hat{\sigma}_{\gamma,\alpha,T|X}[\sigma_{T|X}^{(n)}]$ [cf.~Eqs.~(\ref{BA1V}) and (\ref{BA2V})]; set $n$ as $n+1$.
\UNTIL{convergence.}
\STATE {\bf Output:} {A c-q channel $\sigma_{T|X}^{(n+1)}$}
\end{algorithmic} 
\end{algorithm}
 
As mentioned, when $\gamma$ satisfies the condition (A1) in all iteration steps, i.e., 
when $\gamma$ is sufficiently large, 
Theorem \ref{LBCOV} guarantees the monotonicity of the information bottleneck function:
\begin{align}
f_{\alpha}(\sigma_{T|X}^{(n+1)})  \le
J_{\gamma,\alpha}( \sigma_{T|X}^{(n+1)},\sigma_{T|X}^{(n)}) \le
f_{\alpha}(\sigma_{T|X}^{(n)}) .\Label{BPXV}
\end{align}
Since $f_{\alpha}$ consists of bounded entropic quantities (assuming the system to be finite), it is a bounded quantity. Therefore, the sequence $\{f_{\alpha}(\sigma_{T|X}^{(n)})\}$ in our Algorithm converges. In addition, we can show that the sequence of c-q channels $\{\sigma_{T|X}^{(n)}\}$ converges as well:

\begin{theo}\label{TH3}
When $\gamma \ge 
\sup_{\sigma_{T|X},\sigma_{T|X}'}\gamma(\sigma_{T|X},\sigma_{T|X}')$, 
the sequence $\{\sigma_{T|X}^{(n)}\}$ converges. 
\end{theo}
In particular, since $\alpha \ge \sup_{\sigma_{T|X},\sigma_{T|X}'}\gamma(\sigma_{T|X},\sigma_{T|X}')$, 
the sequence $\{\sigma_{T|X}^{(n)}\}$ converges with $\gamma=\alpha$.

\begin{proof}
Since $\{f_{\alpha}(\sigma_{T|X}^{(n)}) \}$
is monotonically decreasing for $n$,
we have 
\begin{align}
\lim_{n \to \infty}f_{\alpha}(\sigma_{T|X}^{(n)}) -
f_{\alpha}(\sigma_{T|X}^{(n+1)})  =0\Label{AAS}.
\end{align}
Using \eqref{ACO}, we have
  \begin{align}
&f_{\alpha}(\sigma_{T|X}^{(n)}) 
=J_{\gamma,\alpha}( \sigma_{T|X}^{(n)},\sigma_{T|X}^{(n)})  \nonumber \\
=& \gamma \sum_{x}P_X(x)
 D(\sigma_{T|x}^{(n)}\| 
\sigma_{T|x}^{(n+1)} )
+
J_{\gamma,\alpha}( \sigma_{T|X}^{(n+1)},\sigma_{T|X}^{(n)}) \nonumber \\
\ge & 
\gamma \sum_{x}P_X(x)
D(\sigma_{T|x}^{(n)}\| 
\sigma_{T|x}^{(n+1)} )
+f_{\alpha}(\sigma_{T|X}^{(n+1)}) .
\end{align}
Thus, we have
\begin{align}
 \gamma \sum_{x}P_X(x)
 D(\sigma_{T|x}^{(n)}\| 
\sigma_{T|x}^{(n+1)} ) 
\le 
f_{\alpha}(\sigma_{T|X}^{(n)}) -
f_{\alpha}(\sigma_{T|X}^{(n+1)})  \Label{AAS3}.
\end{align}
Since
due to \eqref{AAS} and \eqref{AAS3},
the sequence $\{\sigma_{T|X}^{(n)}\}$ is a Cauchy sequence, it converges. 
\end{proof}

We remark that it is free to choose the convergence criterion in Algorithm \ref{protocol1V}. 


In Algorithm \ref{protocol1V}, $\gamma$ is fixed to be a large enough value. 
Intuitively (see the next paragraph for more detailed discussion), $\gamma$ (or, more precisely, $1/\gamma$) is an acceleration parameter that makes the algorithm converge faster if chosen to be a smaller value.

To begin with, we show the role of $\gamma$ in convergence of the algorithm. Denote by $\sigma_{T|X}^*$  
the convergence point 
of $\{\sigma_{T|X}^{(n)}\}$. 
The performance of our algorithm can be characterized {by} the 
decreasing speed of 
the average divergence between $\sigma_{T|X}^*$ and 
$\sigma_{T|X}^{(n)}$, which is evaluated as
\begin{align}
&\sum_{x}P_X(x)D( \sigma_{T|x}^*\| \sigma_{T|x}^{(n)} ) 
-\sum_{x}P_X(x)D( \sigma_{T|x}^*\| \sigma_{T|x}^{(n+1)} )  \nonumber \\
=&
\sum_{x} P_X(x) \Tr \sigma_{T|x}^*
\Big(\log \sigma_{T|x}^* - \log \sigma_{T|x}^{(n)}\Big)  \nonumber \\
&-
\sum_{x} P_X(x) \Tr \sigma_{T|x}^*
 \Big(\log \sigma_{T|x}^* -\log \sigma_{T|x}^{(n+1)} \Big)
 \nonumber \\
=&\sum_{x} 
P_X(x) \Tr \sigma_{T|x}^*
 \Big(\log \sigma_{T|x}^{(n+1)}- \log \sigma_{T|x}^{(n)}\Big)  \nonumber \\
\stackrel{(a)}{=}&
\sum_{x} 
P_X(x) \Tr \sigma_{T|x}^*
\Big(-\frac{1}{\gamma}
{\cal F}_\alpha[\sigma_{T|X}^{(n)}](x) 
- \log\hat{\eta}_{\gamma,\alpha}[\sigma_{T|X}^{(n)}](x)
\Big) \nonumber \\
\stackrel{(b)}{=}& 
\frac{1}{\gamma}
J_{\gamma,\alpha}( \sigma_{T|X}^{(n+1)},\sigma_{T|X}^{(n)}) 
-{\frac{1}{\gamma}}
\sum_{x} 
P_X(x) \Tr \sigma_{T|x}^*
{\cal F}_\alpha[\sigma_{T|X}^{(n)}](x) 
 \nonumber \\
\stackrel{(c)}{=}& 
\frac{1}{\gamma}\Big(
(J_{\gamma,\alpha}( \sigma_{T|X}^{(n+1)},\sigma_{T|X}^{(n)}) 
-f_{\alpha}(\sigma_{T|X}^*)) \nonumber \\
&+
\sum_{x} 
P_X(x) \Tr \sigma_{T|x}^*
\Big( {\cal F}_\alpha[\sigma_{T|X}^*](x)
- {\cal F}_\alpha[\sigma_{T|X}^{(n)}](x) \Big) \Big),
\Label{NMO1} 
\end{align}
where $(a)$, $(b)$, and $(c)$ follow from 
the combination of \eqref{BA1V} and \eqref{BA2V}, 
\eqref{ACO}, 
and \eqref{BCOV2},
respectively.

The above discussion manifests that if $
\frac{1}{\gamma}
\bigg((
J_{\gamma,\alpha}( \sigma_{T|X}^{(n+1)},\sigma_{T|X}^{(n)}) 
-f_{\alpha}(\sigma_{T|X}^*))+
 \sum_{x} P_X(x) \Tr \sigma_{T|x}^* \Big( {\cal F}_\alpha[\sigma_{T|X}^*](x)
- {\cal F}_\alpha[\sigma_{T|X}^{(n)}](x) \Big)  \bigg)>0$,
making $\gamma$ smaller makes 
the average divergence between $\sigma_{T|X}^*$ and $\sigma_{T|X}^{(n)}$ decrease faster.
On the other hand, making $\gamma$ too small leads to a risk of violating the condition \eqref{XLOV} (and, consequently, breaking the monotonicity of $J_{\gamma,\alpha}$).

\begin{rema}
The reference \cite[Section III]{Ramakrishnan} considered a general setting.
If $\sigma_{T|X}$ is a single density matrix, our method can be considered as a special case of their setting.
However, since $\sigma_{T|X}$ is classical-quantum channel in our case,
our analysis is not a special case of 
their setting.
\end{rema}

\begin{rema}
The references \cite[Appendix A]{PhysRevA.94.012338}
\cite[Appendix A]{8513885} considered the case when the systems $X,Y,T$
are quantum systems and $\alpha=1$.
They derived a necessary condition
for the solution of the minimization problem 
by using Lagrange multiplier method in the same way as \cite{1054753,1054855}.
Using the obtained condition, they 
\cite[Appendix C]{PhysRevA.94.012338}
\cite[Appendix C]{8513885}
also proposed an iterative algorithm to find a solution to satisfy the necessary condition.
It seems that their necessary condition is the same as 
\eqref{AAV} with $\gamma=\alpha=1$.
However, they did not discuss the convergence to a local minimizer in their algorithm.
\end{rema}

\subsection{Numerics on the effects of different $\gamma$}\Label{sec-var-gamma}
To see the effect of different $\gamma$, let us take a look at a concrete example: 
Consider a single-qubit quantum system $Y$ and a classical register $X$ with size $2^8$.
Then, we assume that
$P_X$ is the uniform distribution over ${\cal X}=\{0, \ldots, 2^8-1\}$, and
the density $\rho_{Y|x}$ is given as 
$\rho_{Y|x}=\rho(\theta_{x},\lambda_{x})$, where
\begin{align}
\rho(\theta,\lambda)&:=\exp\left(i\theta\sigma_x\right)\left(\begin{matrix}1-\lambda & 0 \\0 & \lambda\end{matrix}\right)\exp\left(-i\theta\sigma_x\right),
\label{numerics-rhox}
\end{align}
where $\sigma_x=\left(\begin{matrix}0 & 1 \\1 & 0\end{matrix}\right)$ is the Pauli-$X$ matrix.
The parameters $\theta_{x}$ and $\lambda_{x}$ are randomly chosen.

Then, the ensemble we consider admits the following joint density matrix:
\begin{align}\Label{model-select-state-post-measurementB}
\hat{\rho}_{XY}&=\sum_{x}P_{X}(x)|\pi(x)\>\<\pi(x)|\otimes 
\rho\left(\theta_{x},\lambda_{x}\right)
\end{align} 
with $\rho\left(\theta_{x},\lambda_{x}\right)$ given by Eq.~(\ref{numerics-rhox}).

Now, we apply our QIB algorithm (i.e., Algorithm \ref{protocol1V}) to the ensemble (\ref{model-select-state-post-measurementB}). {We consider a classical $T$ whose size is the square root of $|\cal{X}|$ (i.e., $|{\cal T}|=2^4$).} We set $\alpha=1$, and $\beta=10$. Our focus will be the effects of different choices of the acceleration parameter $\gamma$. As shown in Fig.~\ref{fig-IB-gamma}, the choice of $\gamma$ is crucial for the performance, more specifically, the efficiency and the convergence, of the QIB algorithm.

Two interesting phenomena are manifested by our numerics: For one thing, choosing a smaller $\gamma$ will accelerate the course of convergence. As shown in Fig.~\ref{fig-IB-gamma}, by choosing a suitably smaller value of $\gamma$ (e.g., $0.8$ or $0.5$), our QIB algorithm achieves convergence faster than the existing QIB algorithm \cite{PhysRevA.94.012338,8513885}, which corresponds to Algorithm \ref{protocol1V} with $\gamma=1$. For the other, choosing a too small $\gamma$ will ruin the convergence property of the QIB algorithm. For instance, when $\gamma$ is chosen to be $0.4$, $f_{\alpha}$ jumps up after a few iterations and ends up in a much larger value than its initial value. 

\begin{figure}[H]
\begin{center} 
\includegraphics[width=\linewidth]{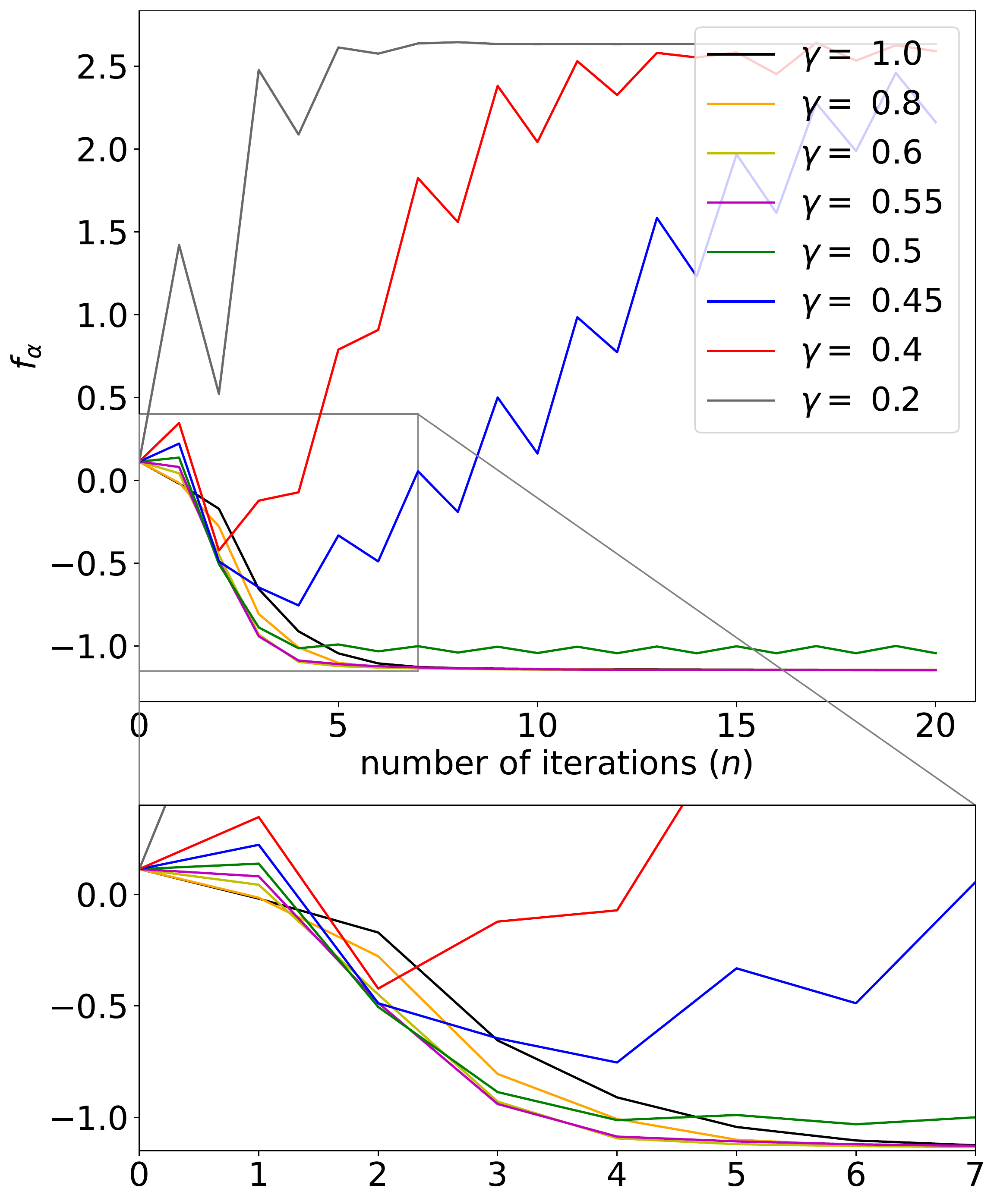} 
\end{center}
\caption{{\bf Performance of Algorithm \ref{protocol1V} for different $\gamma$.}  We apply Algorithm \ref{protocol1V} 
($|{\cal T}|=16$, $\alpha=1$, and $\beta=10$) 
to the joint state (\ref{model-select-state-post-measurementB}). 
The information bottleneck $f_{\alpha}$ is plotted as a function of the number of iterations for different values of $\gamma$. 
The green curve with $\gamma=0.55$ converges most quickly.
It significantly improves the convergence speed in comparison with the black line with $\gamma=1$.
The blue curve with $\gamma=0.45$ goes down even faster in the beginning but gets overtaken after a few iterations. 
Finally, it goes up around $n=7$.
It shows that $\gamma=0.45$ does not satisfy the condition (A1) for $n \ge 7$.
}\Label{fig-IB-gamma} 
\end{figure}

In conclusion, the numerics has justified our theoretical analysis (see Section \ref{S2-C}) 
on the importance of choosing a suitable $\gamma$. We emphasize that our contribution in this direction is twofold:
\begin{enumerate}
    \item We proposed a method of accelerating the QIB algorithm, making it converge within fewer rounds of iteration, by introducing a new parameter $\gamma$ and setting it to be smaller than one.
    \item We showed that the QIB algorithm cannot achieve the desired minimal value of $f_{\alpha}$ if $\gamma$ is too small.  
\end{enumerate}

\subsection{Choice of $\beta$}\Label{sec-beta}
The output of our QIB algorithms depend not only on $\rho_{XY}$ [cf.~(\ref{joint-state})] but also on the choice of $\alpha$ and $\beta$. Intuitively, a larger $\beta$ improves the faithfulness (as it makes $I(Y:T)$ more significant in $f_{\alpha}$), while a smaller $\beta$ leads to more compression (as it makes $I(X:T)$ more significant in $f_{\alpha}$).
Somehow surprisingly, the choice of $\beta$ is not completely free:
In the following, we show that the QIB algorithm will yield a trivial $\sigma_{T|X}$ if $\beta$ is too small.

To consider the relation between the choice of $\beta$ and the resultant information on $T$, we introduce the following condition for a subset ${\cal S} \subset {\cal S}_{X\to T}$,
where ${\cal S}_{X\to T}$ is the the set of all c-q channels from $X$ to $T$, i.e., the set 
$\{ \sigma_{T|X}=(\sigma_{T|x})_{x \in {\cal X}}\}$:
\begin{description}
\item[(A2)] 
For any two distinct elements $\sigma_{T|X}, \sigma_{T|X}'\in
{\cal S}$, $\sum_{x} P_X(x) \Tr_T \sigma_{T|x} 
({\cal F}_\alpha[\sigma_{T|X}](x) -{\cal F}_\alpha[\sigma_{T|X}'](x))> 0$
\end{description}

The condition (A2) is unitarily invariant,
i.e., 
the pair $(\sigma_{T|X},\sigma_{T|X}')$
satisfies the condition (A2),
if and only if the pair $(U\sigma_{T|X} U^\dagger,
U\sigma_{T|X}'U^\dagger)$
satisfies the condition (A2) for any unitary $U$ on $T$. 
Hence, we choose ${\cal S}$ as a unitarily invariant subset.

\begin{theo}\Label{NNK}
Assume that a unitarily invariant subset ${\cal S}$
satisfies (A2). 
Let $\sigma_{T|X}^M :=
\argmin_{\sigma_{T|X} }f_{\alpha}(\sigma_{T|X})$ be the solution to the QIB problem.
When $\sigma_{T|X}^M $ belongs to ${\cal S}$,
$\sigma_{T|x}^M$ is the maximally mixed state on $T$ for any $x$.
\end{theo}

If $\sigma_{T|x}^M$ is the maximally mixed state for every $x$, $T$ is uncorrelated with $Y$ and does not contain any meaningful information.
{In other words},
when the assumption for Theorem \ref{NNK} holds,
the solution of the QIB problem is not useful.
Hence, we need to choose the parameters $\alpha, \beta$
such that condition (A2) does not hold.

Now we discuss how to avoid the condition (A2).
The LHS of (A2) is evaluated as
\begin{widetext}
\begin{align}
&\sum_{x} P_X(x) \Tr_T \sigma_{T|x} 
({\cal F}_\alpha[\sigma_{T|X}](x) -{\cal F}_\alpha[\sigma_{T|X}'](x)) \nonumber \\
=& \Tr_{TY} \sum_{x}P_X(x) (\sigma_{T|x} \otimes \rho_{Y|x})
\Big(
-(\log \sigma_{T}[\sigma_{T|X}]-\log \sigma_{T}[\sigma_{T|X}'])
+\alpha
(\log \sigma_{T|x}- \log \sigma_{T|x}')  \nonumber \\
&+\beta \Big(
(\log  (\sigma_T[\sigma_{T|X}] \otimes \rho_Y )
-\log  (\sigma_T[\sigma_{T|X}'] \otimes \rho_Y ))
-( \log \sigma_{YT} [\sigma_{T|X}]
- \log \sigma_{YT} [\sigma_{T|X}'])
\Big)
\Big)  \nonumber \\
=& \Tr_{TY} \sum_{x}P_X(x) (\sigma_{T|x} \otimes \rho_{Y|x})
\Big(
-(\log \sigma_{T}[\sigma_{T|X}]-\log \sigma_{T}[\sigma_{T|X}'])
+\alpha
(\log P_X(x)\sigma_{T|x}- \log P_X(x)\sigma_{T|x}')  \nonumber \\
&+\beta \Big(
(\log  (\sigma_T[\sigma_{T|X}] \otimes \rho_Y )
-\log  (\sigma_T[\sigma_{T|X}'] \otimes \rho_Y ))
-( \log \sigma_{YT} [\sigma_{T|X}]
- \log \sigma_{YT} [\sigma_{T|X}'])
\Big)
\Big)  \nonumber \\
=& -D(\sigma_{T}[\sigma_{T|X}]\|\sigma_{T}[\sigma_{T|X}']) 
+\alpha D(\sigma_{XT}[\sigma_{T|X}] \| \sigma_{XT}[\sigma_{T|X}'])  \nonumber \\
&- \beta \big(D(\sigma_{YT} [\sigma_{T|X}]\| \sigma_{YT} [\sigma_{T|X}'])
-D(\sigma_{T}[\sigma_{T|X}]\|\sigma_{T}[\sigma_{T|X}'])\big),
\end{align}
\end{widetext}
where $\sigma_{XT}[\sigma_{T|X}]:=\sum_{x}P_X(x)\sigma_{T|x}[\sigma_{T|X}]\otimes |x\rangle \langle x|$.
Since 
$D(\sigma_{YT} [\sigma_{T|X}]\| \sigma_{YT} [\sigma_{T|X}'])
\ge D(\sigma_{T}[\sigma_{T|X}]\|\sigma_{T}[\sigma_{T|X}'])$,
the coefficient of $\beta$ is a negative value.
Hence, {a} smaller $\beta$ has a possibility to satisfy the condition (A2).
That is,
to obtain a useful solution, we need to choose $\beta$ to be a sufficiently large value.

\begin{proofof}{Theorem \ref{NNK}}
Let $U$ be an arbitrary unitary on $\mathcal{T}$.
We define $\sigma_{T|X}^{M'}$ by $\sigma_{T|x}^{M'}=U\sigma_{T|x}^M U^\dagger$.
Substituting {$ \sigma_{T|x}^{(n)}$ with $ \sigma_{T|x}^{M'} $}
{in} \eqref{NMO1}, 
we have 
\begin{align}
0=&\sum_{x}P_X(x)D( \sigma_{T|x}^M\| \sigma_{T|x}^{M'} )  \nonumber \\
&-\sum_{x}P_X(x)D( \sigma_{T|x}^M\| \hat{\sigma}_{\gamma,\alpha,T|x}[\sigma_{T|X}^{M'}] )  \nonumber \\
=& 
\frac{1}{\gamma}
(f_{\alpha}(\sigma_{T|X}^{M'})
-f_{\alpha}(\sigma_{T|X}^M)) \nonumber \\
&+
\frac{1}{\gamma} 
\sum_{x} 
P_X(x) \Tr \sigma_{T|x}^M
\Big( {\cal F}_\alpha[\sigma_{T|X}^M](x)
- {\cal F}_\alpha[\sigma_{T|X}^{M'}](x) \Big) 
\Label{NMO} \\
=&
\frac{1}{\gamma} 
\sum_{x} 
P_X(x) \Tr \sigma_{T|x}^M
\Big( {\cal F}_\alpha[\sigma_{T|X}^M](x)
- {\cal F}_\alpha[\sigma_{T|X}^{M'}](x) \Big) 
.\Label{NMT}
\end{align}
Thus, the condition (A2) implies $\sigma_{T|X}^M=\sigma_{T|X}^{M'}$.
$\sigma_{T|x}^M$ 
is the completely mixed state on $T$ for any $x$.
\end{proofof}

\section{Classical system $T$}\Label{SIV}
Next, we consider the case when $T$ is \emph{constrained} to be  a classical system. We stress that this is a different minimization from the previously discussed one with a quantum 
system $T$, whose minimum may not be attainable with a classical $T$.
Instead, our objective function now is
\begin{align}
{\cal I}_{\alpha,\beta}^c:= \min_{\sigma_{T|X}:diagonal}f_{\alpha}(\sigma_{T|X}).
\end{align} 
Therefore, we need to re-examine the validity of our previous analyses. 

Let us start with the form of QIB algorithm. Fortunately, our algorithm with a quantum 
system $T$ can be applied to this case, simply with the adaptation that the states $\sigma_{T|x}$ are limited to 
diagonal density matrices with respect to the basis $\{|t\rangle\}$
of $T$. 
Under this condition, the states
$\hat{\sigma}_{\gamma,\alpha,T|x}[\sigma_{T|X}]$ are also 
diagonal density matrices.
Therefore, when we set the initial state as diagonal density matrices, 
Algorithm \ref{protocol1V} works for this case.

The above discussion leads to an interesting observation as follows. 
The convergent $\sigma_{T|X}^*$ with initial diagonal $\sigma_{T|X}$ satisfies the condition \eqref{ACAC}
and it is also diagonal.
That is, if the minimum with classical $T$ is strictly larger 
{than} the minimum with quantum $T$,
the minimum with classical $T$ is {an example} for the following statement:
A solution of the condition \eqref{ACAC} does not necessarily give the minimum of 
$f_{\alpha}$ with quantum $T$.
This fact shows the possible risk that a solution to \eqref{ACAC} might be a saddle point or a local minimum rather than the global minimum for $f_{\alpha}$ with quantum $T$. 

When the states $\sigma_{T|x}$ are limited to 
diagonal density matrices with respect to the basis $\{|t\rangle\}$
of $T$, 
$\sigma_{TY} [\sigma_{T|X}]$ is commutative with $ \sigma_{T} [\sigma_{T|X}]$
so that
we can define 
$\sigma_{Y|T} [\sigma_{T|X}]:=\sigma_{TY} [\sigma_{T|X}] \sigma_{T} [\sigma_{T|X}]^{-1}$.
Then,
$\hat{\sigma}_{\gamma,\alpha,T}[\sigma_{T|X}](x)$
is simplified as follows.
\begin{align}
&\log \hat{\sigma}_{\gamma,\alpha,T}[\sigma_{T|X}](x) \nonumber \\
=&(1-\frac{\alpha}{\gamma})\log \sigma_{T|x}
+\frac{1}{\gamma}\log \sigma_{T}[\sigma_{T|X}]
 \nonumber \\
&-\frac{\beta}{\gamma}
\Tr_Y \Big( 
\rho_{Y|x} 
(\log \rho_Y
 -\log \sigma_{Y|T} [\sigma_{T|X}]
  )\Big) .
\end{align}


The notion of unitary invariance is reduced to invariance under permutations on $T$, and the condition (A2) is invariant under permutations on $T$. 
Then, Theorem \ref{NNK} can be rewritten as follows.
\begin{theo}\Label{NNKC}
Assume that a subset ${\cal S}$
satisfies (A2) and 
is invariant under any permutation on $T$. 
Let $\sigma_{T|X}^* $ be the minimizer of $ 
\min_{\sigma_{T|X}:diagonal }f_{\alpha}(\sigma_{T|X})$.
When $\sigma_{T|X}^* $ belongs to ${\cal S}$,
$\sigma_{T|x}^*$ is the 
uniform distribution over $T$ for any $x$. 
\end{theo}
Theorem \ref{NNKC} can be shown in the same way as Theorem \ref{NNK}.

In this case, we can make a more precise discussion for the condition (A2).
For this purpose,
we consider the maximum ratio
\begin{align}
\kappa:=\max_{Q_X,Q_X'}
\frac{D( \sum_{x}Q_X(x) \rho_{Y|x} \| \sum_{x}Q_X'(x) \rho_{Y|x})}
{D(Q_X\|Q_X')}.
\end{align}
The inequality $\kappa\le 1$ follows from the information processing inequality
for the map $Q_X \mapsto  \sum_{x}Q_X(x) \rho_{Y|x}$.
In this condition, 
$\sigma_{T}[\sigma_{T|X}]
$ is written as $\sum_{t}Q_{T}[ \sigma_{T|X}](t)|t\rangle \langle t|$ by 
using a distribution $Q_{T}[ \sigma_{T|X}]$.
Then,
the LHS of (A2) is simplified as
\begin{align}
&\sum_{x} P_X(x) \Tr_T \sigma_{T|x}  
({\cal F}_\alpha[\sigma_{T|X}](x) -{\cal F}_\alpha[\sigma_{T|X}'](x))  \nonumber \\
=&(\beta -1)D(\sigma_{T}[\sigma_{T|X}]\|\sigma_{T}[\sigma_{T|X}']) \nonumber \\
&+\alpha D(\sigma_{XT}[\sigma_{T|X}] \| \sigma_{XT}[\sigma_{T|X}'])  \nonumber \\
&- \beta D(\sigma_{YT} [\sigma_{T|X}]
\| \sigma_{YT} [\sigma_{T|X}'])  \nonumber \\
=&(\alpha -1) D(\sigma_{T}[\sigma_{T|X}]\|\sigma_{T}[\sigma_{T|X}'])  \nonumber \\
&+\sum_{t} Q_{T}[ \sigma_{T|X}] (t)
\Big(\alpha  D(\sigma_{X|T=t}[\sigma_{T|X}] \| \sigma_{X|T=t}[\sigma_{T|X}']) \nonumber \\
&- \beta D(\sigma_{Y|T=t} [\sigma_{T|X}] \| \sigma_{Y|T=t} [\sigma_{T|X}']) 
\Big) \nonumber \\
\ge &(\alpha -1) D(\sigma_{T}[\sigma_{T|X}]\|\sigma_{T}[\sigma_{T|X}'])  \nonumber \\
&+
(\alpha-\beta \kappa)
\sum_{t} Q_{T}[ \sigma_{T|X}] (t) \nonumber \\
&\cdot D(\sigma_{X|T=t}[\sigma_{T|X}] \| \sigma_{X|T=t}[\sigma_{T|X}']).
\end{align}
When the condition $\alpha \ge 1,
\frac{\alpha}{{\kappa}} > \beta $ holds,
the LHS of (A2) is positive for $ \sigma_{T|X}\neq \sigma_{T|X}'$.
Hence, to extract useful $\sigma_{T|X} $, we need to choose $\beta$ to satisfy the condition
$\beta > \frac{\alpha}{\kappa} $ with $\alpha=1$.
In fact, even when $\beta > \frac{\alpha}{\kappa} $,
there is a possibility that a permutation-invariant subset ${\cal S}$ satisfies (A2).
Due to Theorem \ref{NNKC},
when a permutation-invariant subset ${\cal S}$ satisfies (A2), 
a useful solution does not belong to the subset ${\cal S}$. 
Hence, to obtain a useful solution, we need to choose $\beta$ sufficiently large beyond 
the above condition $\beta> \frac{\alpha}{\kappa} $ with $\alpha=1$.

\begin{rema}
We consider the case with classical $Y$ and $\gamma=\alpha$.
The operator
$\hat{\sigma}_{\alpha,T}[\sigma_{T|X}](x)$
is simplified as follows.
\begin{align}
&\hat{\sigma}_{\alpha,T}[\sigma_{T|X}](x) \nonumber \\
=&
\exp \Big(\frac{1}{\alpha}\log \sigma_{T}[\sigma_{T|X}]
 \nonumber \\
&-\frac{\beta}{\alpha}
\Tr_Y \Big( \rho_{Y|x} 
(\log \rho_Y 
 -\log \sigma_{Y|T} [\sigma_{T|X}]
 )\Big) \Big).
\end{align}

In this case, 
the reference \cite[(14) Section 3]{10.1162/NECO_a_00961} 
proposed the following update rule:
\begin{align}
\hat{\tau}_{T|x}[\sigma_{T|X}]
:=\frac{1}{\Tr \hat{\tau}_{T}[\sigma_{T|X}](x)}
\hat{\tau}_{T}[\sigma_{T|X}](x), \Label{AAR}
\end{align}
where the operator $\hat{\tau}_{T}[\sigma_{T|X}](x)$ is defined as
\begin{align}
&\hat{\tau}_{T}[\sigma_{T|X}](x) \nonumber \\
:=&
\exp \Big(\frac{1}{\alpha}\log \sigma_{T}[\sigma_{T|X}]
 \nonumber \\
&-\frac{\beta}{\alpha}
\Tr_Y \Big( \rho_{Y|x} 
(\log \rho_{Y|x} 
 -\log \sigma_{Y|T} [\sigma_{T|X}]
 )\Big) \Big).
\end{align}
Since
\begin{align}
\log \hat{\tau}_{T}[\sigma_{T|X}](x)
-\log \hat{\sigma}_{T}[\sigma_{T|X}](x) 
=\frac{\beta}{\alpha} D (\rho_{Y|x}\| \rho_{Y}) ,
\end{align}
we have
\begin{align}
&\hat{\tau}_{T|x}[\sigma_{T|X}]  \nonumber \\
=&\frac{1}{\Tr e^{\frac{\beta}{\alpha} D (\rho_{Y|x}\| \rho_{Y})}\hat{\sigma}_{T}[\sigma_{T|X}](x)
}
e^{\frac{\beta}{\alpha} D (\rho_{Y|x}\| \rho_{Y})}\hat{\sigma}_{T}[\sigma_{T|X}](x) \nonumber \\
=&\hat{\sigma}_{T|x}[\sigma_{T|X}] .
\end{align}
That is, the update rule \eqref{AAR} by 
\cite[(14) Section 3]{10.1162/NECO_a_00961} 
is the same as ours of this special case. 
In particular, the update rule \eqref{AAR} with $\alpha=1$
coincides with the update rule by
the reference \cite{tishby1999}.
\end{rema}

\begin{rema}
When the system $Y$ is classical and $\alpha=1$, 
the reference \cite[Appendix B]{PhysRevA.94.012338}
claimed that there is no difference between 
 the optimal value with quantum $T$ and 
 the optimal value with classical $T$.
{Since their algorithm works with $T$ of a fixed size,
it can be considered that 
they claimed the above statement 
when the size of $T$ is fixed.}
{However, their proof (see \cite[Appendix B II]{PhysRevA.94.012338}) contains a gap:
The statement under Eq.~(B23) that ``the Lagrangian is invariant under a measurement of the memory $M$ in a chosen basis $|m\>$'' is not backed by a rigorous mathematical proof.
It is thus unclear whether this statement and, consequently, the claim that there is no quantum advantage are correct.}
On the other hand, as we show next, the optimal value with quantum $T$ can be strictly smaller than the optimal value with classical $T$.
{That is, the claim in \cite[Appendix B]{PhysRevA.94.012338}
contradicts with our result of the next section.}
\end{rema}

\section{Quantum advantage for $T$}\Label{S4}
To see the advantage of quantum system $T$ over classical system $T$,
we discuss several examples with the strict inequality
\begin{align}
{\cal I}_{\alpha,\beta}<  {\cal I}_{\alpha,\beta}^c.\Label{XMP}
\end{align}
We provide an analytical example in this section
 and a numerical example with application in quantum machine learning in Section \ref{S8-B}
{when the size of the system $T$ is fixed.
Generally, to achieve the optimal performance, 
we need to choose the system $T$ as a sufficiently large dimensional system.
However, in this section,
to provide analytical examples,
we fix the size of the system $T$ to a certain value.}


Assume that ${\cal Y}$ is a classical system of size $d$.
The size of ${\cal X}$ is $k$ times of the size $d$ of ${\cal Y}$.
We assume that ${\cal X}$ is given as ${\cal X}_1 \times {\cal X}_2$ 
with ${\cal X}_1={\cal Y}$
and $|{\cal X}_2|=k$.
The distribution of $X$ is assumed to be uniform.
We focus on the quantum system $T$ with the dimension $n< d$.

\begin{lemm}
When $\beta \ge 1 $ and $ \beta\ge  \alpha$, we have
\begin{align}\Label{bound-quantum}
{\cal I}_{\alpha,\beta}= (1-\beta)\log n
\end{align}
\end{lemm}

\begin{proof}
First, we show a bound on the QIB for generic (quantum) $T$.
For any $\sigma_{T|x}$, we have 
$H(T)\ge I(T:X) \ge I(T:Y) $.
Hence, the relation $\beta-\alpha \ge 0$ implies
$-(\beta-\alpha) I(T:Y)\ge -(\beta-\alpha) H(T) $.
Hence, we have
\begin{align}
&f_\alpha(\sigma_{T|x})
= (1-\alpha)H(T)
+\alpha I(T:X)
-\beta I(T:Y) \nonumber \\
\ge &
(1-\alpha)H(T)
-(\beta-\alpha) I(T:Y)
\ge  (1-\beta)H(T).\Label{AMR}
\end{align}
Since $H(T) \le \log n  $ and $1-\beta \le 0$, we obtain 
\begin{align} 
{\cal I}_{\alpha,\beta}\ge (1-\beta)\log n.
\end{align}
The above bound is tight. Indeed, we choose $\sigma_{T|x_1,x_2}$ as the pure state
$\sum_{t=1}^n\frac{1}{\sqrt{n}}e^{ \frac{2\pi x_1}{n} i }|t\rangle$.
Then, we have 
$H(T)= \log n$.
Also, $H(T)=I(T:X)=I(T:Y) $. Therefore, $f_{\alpha}(\sigma_{T|x})=(1-\beta)\log n$.
\end{proof}

Next, we focus on the case when $T$ is a classical system of dimension $n< d$. 

\begin{lemm}\label{L7}
Assume that $d=mn+l $ with $0\le l < n$.
When $\beta \ge 1\ge \alpha $, we have
\begin{align}\Label{bound-classical}
{\cal I}_{\alpha,\beta}^c= 
(1-\beta) \Big(\frac{l(m+1)}{d} \log \frac{d}{m+1}
+ \frac{(n-l)m}{d}\log \frac{d}{m}\Big)
\end{align}
\end{lemm}

\begin{proof}
Any channel $\sigma_{T|x}$ can be written as a probabilistic mixture of deterministic channels
$\sigma_{T|x}^j $.
That is, we have
\begin{align}
\sigma_{T|x}=\sum_j p_j \sigma_{T|x}^j.
\end{align}
Since $Y$ is independent of $X_2$ and 
the random variable $J$ describing the choice of $j$,
we have 
\begin{align}
I(T:Y|JX_2) 
=&
I(T:Y|JX_2)+I(Y:JX_2) \nonumber\\
=&I(TJX_2:Y) \ge I(T:Y) . \Label{AM1}
\end{align}
Also, we have
\begin{align}
H(T) \ge H(T|J X_2).\Label{AM2}
\end{align}
Then, we have
\begin{align}
&f_\alpha(\sigma_{T|x}) \stackrel{(a)}{\ge}
(1-\alpha)H(T)
-(\beta-\alpha) I(T:Y) \nonumber \\
\stackrel{(b)}{\ge}& 
(1-\alpha)H(T|J X_2)
-(\beta-\alpha) I(T:Y|J X_2) ,\Label{CME}
\end{align}
where $(a)$ follows from \eqref{AMR}, and $(b)$ follows from 
\eqref{AM1} and \eqref{AM2}.
The minimization of
$(1-\alpha)H(T|J X_2)
-(\beta-\alpha) I(T:Y|J X_2) $
equals the minimization
of the same function under the condition that
$\sigma_{T|X}$ is a deterministic channel
and $\sigma_{T|x_1x_2}$ depends only on $x_1$.

Under this condition, we have 
$I(T:X)=I(T:X_1)=I(T:Y)$, which implies the equality in $(a)$ at \eqref{CME}.
Therefore, for the minimization, we can impose this condition, i.e.,
the variable $T$ is determined only by $X_1=Y$, {which implies 
$I(T:Y)=H(T)$.}
In this case, we have $f_\alpha(\sigma_{T|x})= (1-\beta)H(T)$.
In the classical case, the maximum entropy $H(T) $ among deterministic channels
is achieved when the distribution $(P_T(t))_{t=1}^n$ as close as possible to the uniform distribution, i.e., 
$P_T=(\overbrace{\frac{m+1}{d},\ldots,\frac{m+1}{d}}^l, 
\overbrace{\frac{m}{d},\ldots,\frac{m}{d}}^{n-l})$.
Hence, the maximum entropy $H(T) $ is
$\frac{l(m+1)}{d} \log \frac{d}{m+1}+ \frac{(n-l)m}{d}\log \frac{d}{m}$.
Therefore, we obtain the desired statement. 
\end{proof}

{When the conditions of Lemma \ref{L7} hold,
$d$ cannot be divided by $n$.}
In this case, since $\frac{l(m+1)}{d} \log \frac{d}{m+1}
+ \frac{(n-l)m}{d}\log \frac{d}{m}$ is strictly smaller than
$\log n$, when the state $\rho_{XY}$ is close to the state $\sum_{x}\frac{1}{d}|x,x\rangle \langle x,x|$,
the strict inequality \eqref{XMP} holds. There is clearly an advantage of using a quantum $T$.

\section{Quantum feature maps with QIB}\Label{S8}
\subsection{Information bottleneck in supervised learning}
Supervised learning is a cornerstone of machine learning.
Given a dataset $\{(x,y)\}$ sampled from an unknown probability distribution $P_{XY}$, a general supervised learning task is to find a classifier such that, for any testing data $(x',y')$ sampled from the same distribution $P_{XC}$, it predicts the label $y'$ with as high accuracy as possible given $x'$.

Remarkably, recent studies \cite{tishby2015deep,shwartz2017opening,goldfeld2020information} on the information bottleneck theory    showed evidences that the training phase of deep learning can be divided into two stages. In the first stage, a representation $T$ of $X$ that faithfully encodes its correlation with $Y$ is found, featured by increasing $I(T:Y)$. 
In the second stage, the size of $T$ is compressed, featured by decreasing $I(T:X)$. 
This result suggests that finding an efficient and compressed representation of $X$ facilitates data classification.

\begin{figure}[htb]
\centering
\includegraphics[width=0.95\linewidth]{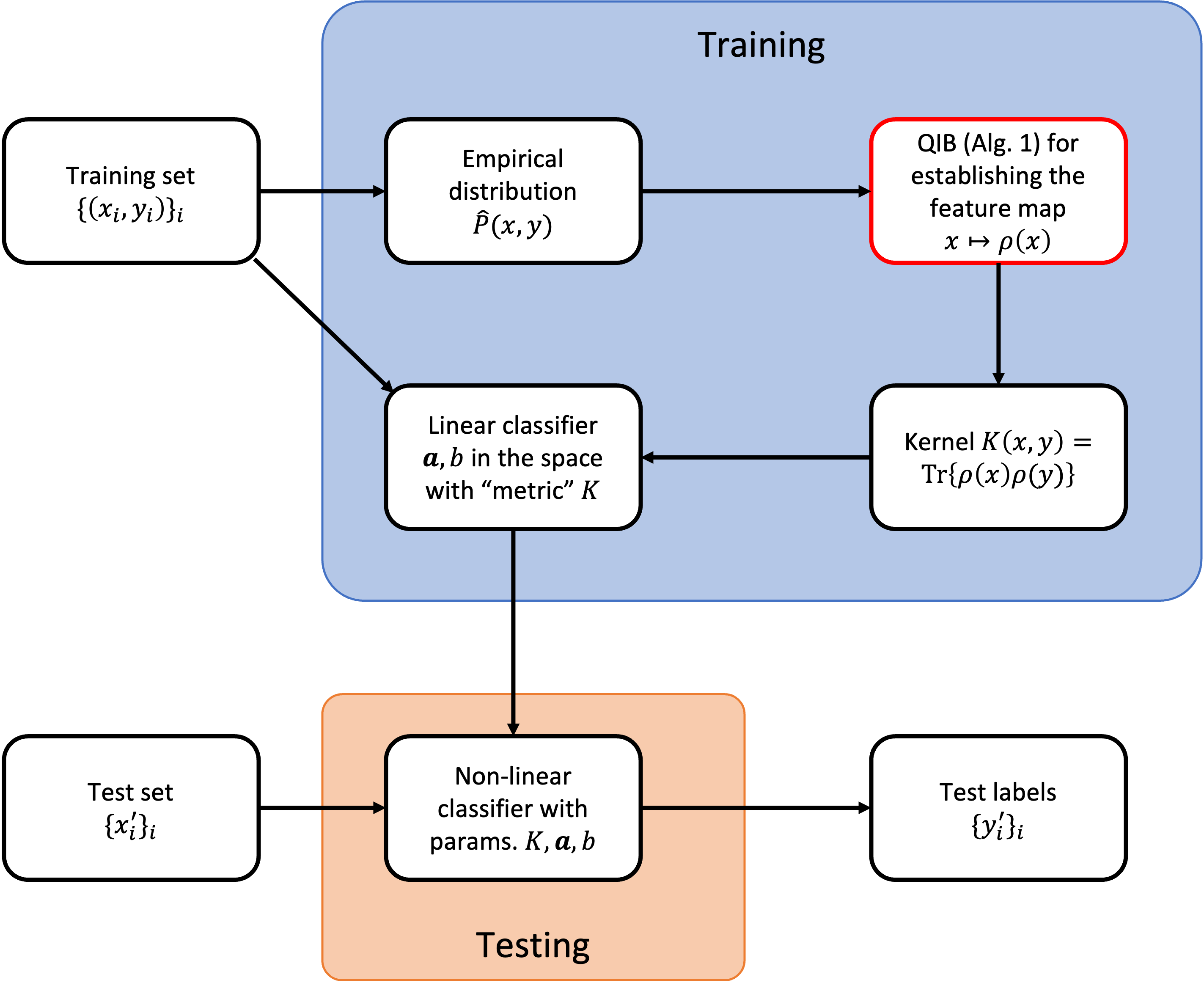}\caption{{\bf Data classification with quantum feature maps.} The flowchart illustrates the training phase and the testing phase of data classification using the technique of quantum feature maps. The part where our QIB algorithm is applied is highlighted.}\Label{fig-qib-classifier} 
\end{figure}

\subsection{Quantum feature maps}\Label{S8-B}
Following the above intuition, we propose a classical-quantum hybrid algorithm of data classification, by combining the QIB algorithm with the kernel method. The idea is illustrated in the flowchart in Fig.~\ref{fig-qib-classifier}.
Given a training dataset $\set{S}_{\rm train}$, the algorithm first identifies an efficient representation $T$ of $X$ by minimising the information bottleneck $f_{\alpha}:=H(T)-\alpha H(T|X)-\beta I(T:Y)$. Then a classifier is constructed that yields a prediction $\hat{Y}$ based on the state in $T$ corresponding to the value of $X$. 
For simplicity, we consider for now the case when $Y\in\{1,-1\}$ is binary. In the first step, we set the representation $T$ to be a quantum state $\rho(x)$ that depends on the data $x$, and we obtain $\rho(x)$ via Algorithm \ref{protocol1V}. 
In the second step, we use a linear classifier 
\begin{align}
c_{\rm QIB}\left(\rho(\tilde{x})\right)={\rm sgn}\left(\Tr[A\rho(\tilde{x})]+b\right)
\end{align}
where $A$ is a Hermitian operator and $b\in\R$. We further consider $A$ that can be expressed as a linear combination $A=\sum_{x:(x,y)\in\set{S}_{\rm train}}a_x \rho(x)$, and the classifier has the reduced form
\begin{align}\Label{qib-classifier}
c_{\rm QIB}\left(\rho(\tilde{x})\right)={\rm sgn}\left(\sum_{x:(x,c)\in\set{S}_{\rm train}}a_x K(x,\tilde{x})+b\right),
\end{align}
where $K(x,\tilde{x})$ is the \emph{kernel} function, in our case given by the Hilbert-Schmidt (HS) inner product of quantum states and can be evaluated by performing the SWAP test on a quantum computer:
\begin{align}\Label{kernel-HS}
K(x,y)=\Tr\{\rho(x)\rho(y)\}.
\end{align} 

The algorithm is summarised as follows:
\begin{algorithm}[H]
\caption{QIB for data classification}
\Label{alg-classification}
\begin{algorithmic}
\STATE {\bf input:} {A training data set $\set{S}_{\rm train}=\{(x,y)\}$}; configuration $(\alpha,\beta,\gamma)$.  
\STATE {\bf input:} {A classifier $c_{\rm QIB}:X\to\hat{Y}$}.
\STATE  {\bf 1)} Generate an empirical distribution $\hat{P}(x,y)$ from $\set{S}_{\rm train}$.
\STATE  {\bf 2)} Run Algorithm \ref{protocol1V} 
with $\hat{P}(x,y)$ as input and certain (adjustable) parameters $\alpha$,$\beta$,$\gamma$.
\STATE  {\bf 3)} Compute the kernel $K$ in Eq.~(\ref{qib-classifier}) using the output of Step 2). 
\STATE  {\bf 4)} Train the classifier (\ref{qib-classifier}) with $\set{S}_{\rm train}$ and output the trained classifier.
\end{algorithmic}
\end{algorithm}

We remark that the quantum kernel method, where a mapping $x\to\rho(x)$ is constructed for better classification, has been a hot topic recently (see, e.g., \cite{schuld2019quantum,havlivcek2019supervised,blank2020quantum,lloyd2020quantum,perez2020data,schuld2021supervised}). The key distinction between existing works and our present method is the following: In existing works, the parameter $x$ is passed to a parametrised (a.k.a. variational) quantum circuit that prepares the state $\rho(x)$. One needs to train the circuit parameters on a quantum computer to obtain a good mapping $x\mapsto\rho(x)$, which is called a feature map. 
In the near term, this method might be subject to the physical limitations of quantum devices. In contrast, in our present method $\rho(x)$ is directly computed via a simple iterative algorithm. Therefore, there are two possible ways of realizing our present method, i.e., Algorithm \ref{alg-classification}. In the near term, we can regard Algorithm \ref{alg-classification} as a ``quantum-inspired'' classical algorithm, and evaluate everything on a classical computer. When large-scale quantum computing becomes feasible, Algorithm \ref{alg-classification} can be readily ``quantised''. 
Indeed, the evaluation of $\rho(x)$ in each iteration requires subroutines that compute matrix powers and logarithm and solve linear systems, which have already been developed in Refs.~\cite{harrow2009quantum,low2019hamiltonian,low2017hamiltonian,gilyen2019quantum}.

\subsection{Numerical experiments}

\begin{figure}[htb]
\begin{center} 
\includegraphics[width=0.95\linewidth]{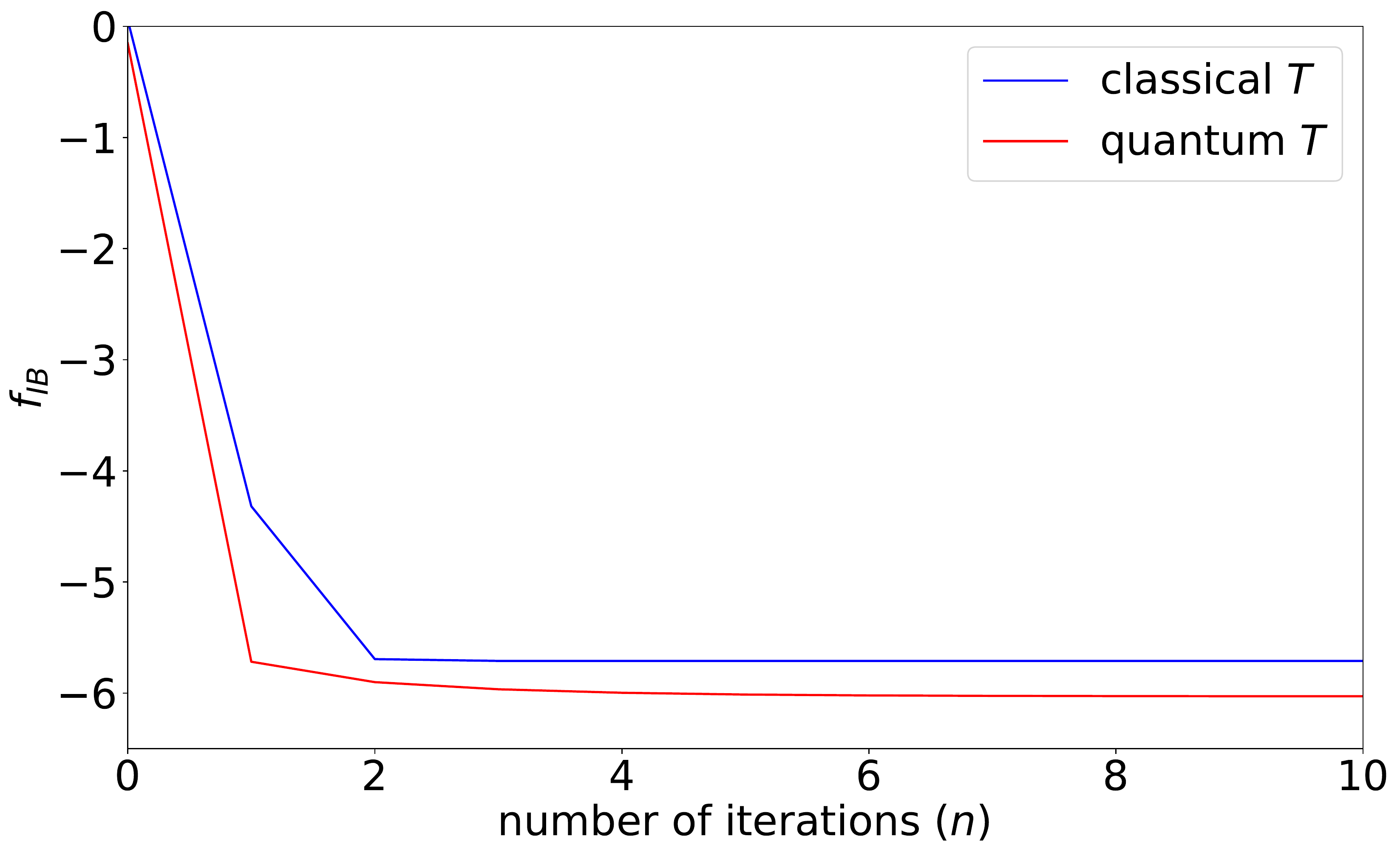} 
\includegraphics[width=0.6\linewidth]{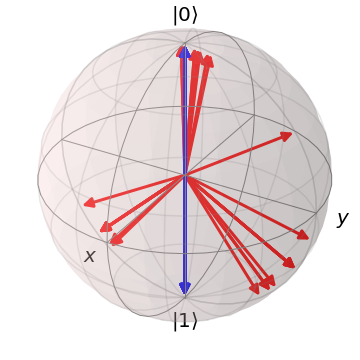} 
\end{center}
\caption{{\bf Quantum vs classical feature maps.}  We run Algorithm \ref{protocol1V} with $\alpha=\gamma=1$, $\beta=15$ on 
the distribution $\tilde{P}_{XY}$ based on the training data, 
and compare the converging values of QIB 
when $T$ is classical (i.e., a probabilistic bit) and when $T$ is quantum (i.e., a single qubit). 
This numerics shows the advantage of use of quantum $T$ over classical $T$.
The final feature maps with quantum $T$ (plotted in red) 
and with the classical-$T$ (plotted in blue) are visualised in the Bloch ball.
}\Label{fig-Q-vs-C-classification} 
\end{figure}
 
As a proof-of-principle experiment, we tested the performance of our QIB classifier on a dataset on $\R^2$, generated in the following way:
First, we define the discrete sets ${\cal X}={\cal X}_1\times {\cal X}_2$ and ${\cal Y}$, 
with ${\cal X}_1={\cal Y}=\{0,1,2\}$ and ${\cal X}_2=\{0,1,\dots,9\}$.
To apply our classification method, 
we arbitrarily choose permutation $\pi$,
and generate $n'=400$ independent and identically distributed data
$(\tilde{X}_{1,i},\tilde{X}_{2,i},Y_i)$ for $i=1, \ldots, n' $ as follows.
We independently generate $(X_{1,i},X_{2,i},Y_i)$ according to 
the following distribution
\begin{align}
P_{XY}(x_1,x_2,y):=P_Y(y)Q_{X_1|Y}(x'_1,y)Q_{X_2|X_1}(x'_2,x'_1),
\end{align}
where $P_Y$ is the uniform distribution over $Y$, 
$Q_{X_1|Y}(x_1,y)=\delta(x_1,y)$, $Q_{X_2|X_1}(x_2,x_1)=\frac{\delta(x_1,x_2)+1}{|{\cal X}_2|+1}$, and $(x'_1,x'_2)=\pi(x_1,x_2)$. 
Next, we generate the random variables 
$\tilde{X}_{j,i}:=X_{j,i}+R_{j,i}$, where 
the random variable $R_{j,i}$ is subject to the uniform distribution in the interval $[0,1.2)$
unless $i=1, X_i=2$ nor $i=2,X_i=9$,
it is subject to the uniform distribution in the interval $[0,1)$
otherwise.
Then, using the obtained data
$(\lfloor\tilde{X}_{1,i}\rfloor,\lfloor\tilde{X}_{2,i}\rfloor,Y_i)$ with $i=1, \ldots, n $,
we define its empirical distribution $\tilde{P}_{XY}$.
We apply Algorithm \ref{protocol1V} to the distribution $\tilde{P}_{XY}$ as Fig. \ref{fig-Q-vs-C-classification}.
In the case with the distribution $\tilde{P}_{XY}$,
 Algorithm \ref{protocol1V} with quantum $T$ can realize a smaller 
 $f_{\alpha}$ than  Algorithm \ref{protocol1V} with classical $T$,
 which shows the advantage of quantum $T$ over classical $T$.

In the classification experiment, $50\%$ of the data are used as the training set and the rest are used as the testing set. The kernel is constructed with Algorithm \ref{alg-classification} with $\alpha=1,\beta=15,\gamma=1$, a single-qubit register $T$, and 10 iterations. 
We consider both when $T$ is a generic qubit system and when $T$ is restricted to a binary classical system, and we compare their performance. As can be seen from Fig.~\ref{fig-Q-vs-C-classification}, the case of quantum $T$ has lower IB value than the case of classical $T$.  The final feature map $\sigma_{T|X}$ for the quantum $T$ case suffers from certain degree of dispersion due to the random noise $r_1,r_2$, but the quantum features still form 3 clusters. In contrast, the final $\sigma_{T|X}$ in the classical $T$ case  maps different values of $X$ into two clusters. 

The effect of the above distinction is made apparent in the classification performance. 
In Fig.~\ref{fig-decision-region}, the performance of the classifiers constructed from the kernels are illustrated via their decision regions. It can be seen that, since the classical-$T$ feature map groups $X$ into two clusters, its resultant classifier gives a binary prediction on any input data, giving up the least possible label. In stark contrast, the quantum-$T$ feature map utilizes the full Bloch ball to generate 3 clusters, leading to a much higher accuracy of prediction. The advantage of a genuinely quantum feature map is thus manifested by this numerical example.

For reference, in Fig.~\ref{fig-decision-region}, we also plot the performance of two standard methods of classical feature maps. The referential methods (linear kernel and polynomial kernel) achieve accuracies (defined by the ratio of correct predictions in the testing set) $0.64$ and $0.62$, which is slightly higher than the classical-$T$ information bottleneck kernel ($0.565$) but much lower than the QIB kernel ($0.92$). This further justifies the superior performance of our QIB method in classification.

\begin{figure}[H]
\begin{center} 
\includegraphics[width=0.95\linewidth]{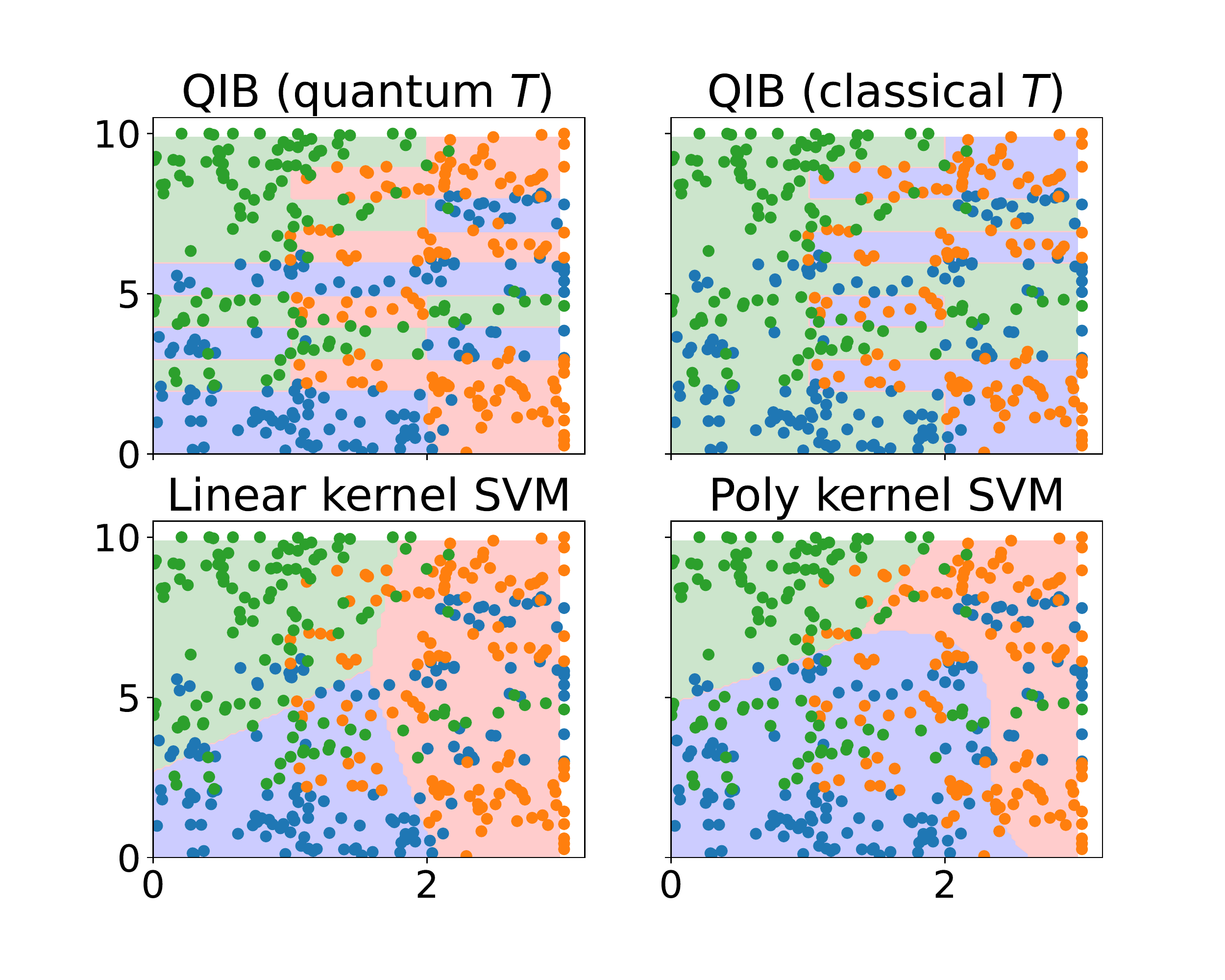}
\end{center}
\caption{{\bf Decision regions of the QIB classifier and reference classifiers.} The decision regions of the QIB classifier, the classical-$T$ IB classifier, and two reference classifiers are plotted together with the test data. The different dot colors correspond to data with different labels, and the color of each region corresponds to the prediction made by the classifier for data in that region.  }
\Label{fig-decision-region}  
\end{figure}

\section{Quantum deterministic information bottleneck (QDIB)}\Label{S6}
Considering the limit $\alpha \to +0$,
the paper \cite{10.1162/NECO_a_00961} proposed  deterministic IB, which minimize $f_0$.
Now, we consider this minimization with quantum systems $T,Y$ and classical system $X$.
First, we define
\begin{align}
&\hat{\sigma}_{0,T|x}[\sigma_{T|X}]  \nonumber \\
:=&
\frac{1}{\Tr \sigma_{T|x} P_{T|x}[\sigma_{T|X}]}
P_{T|x}[\sigma_{T|X}] \sigma_{T|x} P_{T|x}[\sigma_{T|X}],
\end{align}
where $P_{T|x}[\sigma_{T|X}]$ is the projection to the maximum eigenvalue of 
the operator $ (1-\beta)\log \sigma_T[\sigma_{T|X}]
+\beta \Tr_Y \rho_{Y|x} (\log \sigma_{YT}[\sigma_{T|X}]-\log \rho_Y)$.

Given an initial point $\sigma_{T|X}^{(1)}$, we propose the following update rule
\begin{align}
\sigma_{T|X}^{(n+1)}:=\hat{\sigma}_{0,T|X}[\sigma_{T|X}^{(n)}].
\end{align}
As shown below, each step of this algorithm improves the value of the target function $f_0$.

The operator $\hat{\sigma}_{0,T|x}[\sigma_{T|X}]$ is characterized as
\begin{align}
\hat{\sigma}_{0,T|x}[\sigma_{T|X}]=
\lim_{\alpha \to 0}
\hat{\sigma}_{\alpha,\alpha,T|x}[\sigma_{T|X}].
\end{align}
Since
Theorem \ref{LBCOV} and \eqref{AMX} guarantee 
\begin{align}
& f_\alpha(\hat{\sigma}_{\alpha,\alpha,T|X}[\sigma_{T|X}])  \nonumber \\
= & J_{\alpha,\alpha}(
\hat{\sigma}_{\alpha,\alpha,T|X}[\sigma_{T|X}], \hat{\sigma}_{\alpha,\alpha,T|X}[\sigma_{T|X}])  \nonumber  \\
\le & J_{\alpha,\alpha}(
\hat{\sigma}_{\alpha,\alpha,T|X}[\sigma_{T|X}], \sigma_{T|X})  \nonumber \\
\le & J_{\alpha,\alpha}( \sigma_{T|X},\sigma_{T|X})
=f_\alpha(\sigma_{T|X}),\Label{LLP}
\end{align}
the limit $\alpha \to 0$ in \eqref{LLP} implies
\begin{align}
f_{\alpha\to0}( \hat{\sigma}_{0,T|X}[\sigma_{T|X}]) \le f_{\alpha\to0}(\sigma_{T|X}]),
\end{align}
which shows that
each step of this algorithm improves the value of the target function $f_{\rm DIB}:=f_{\alpha\to0}$.

\begin{algorithm}[H]
\caption{Quantum deterministic information bottleneck (QDIB) algorithm}
\Label{protocol-DIB}
\begin{algorithmic}[1]
\STATE {\bf Input:} A joint state $\rho_{XY}$ [cf.~(\ref{joint-state})].
\STATE Create a counter $n$ as the number of iterations, initialized to 1. 
\REPEAT  
\STATE Choose $\sigma_{T|X}^{(n+1)}$ as 
\begin{align}
\sigma_{T|x}^{(n+1)}=\frac{P_{T|x}[\sigma^{(n)}_{T|X}] \sigma^{(n)}_{T|x} P_{T|x}[\sigma^{(n)}_{T|X}]}{\Tr\big(\sigma^{(n)}_{T|x} P_{T|x}[\sigma_{T|X}]\big)}
\end{align}
 where $P_{T|x}[\sigma^{(n)}_{T|X}]$ is the projection on the space spanned by the eigenvectors of ${\cal F}_{\alpha=0}[\sigma^{(n)}_{T|X}](x)$ [cf.~(\ref{F-function})] corresponding to the minimum eigenvalue.
\STATE Set $n$ as $n+1$.
\UNTIL{convergence.}
\STATE {\bf Output:} {A c-q channel $\sigma_{T|X}^{(n+1)}$}
\end{algorithmic}
\end{algorithm}

\section{Approximate sufficient statistics from DIB}\Label{SQ}
\subsection{Task formulation}
Next, we discuss how DIB can be used for the extraction of useful information
under a classical-quantum (c-q) joint system composed of $X$ and $Y$ with the joint state $\rho_{XY}:= \sum_{x}P_X(x)|x\rangle \langle x|\otimes \rho_{Y|x}$,
where $X$ is a classical system and $Y $ is a quantum system.
For example, assume that our interest is in the quantum phenomena in the quantum system $Y$.
This quantum system $Y$ is correlated to the classical system $X$.
However, there is a possibility that the classical system $X$ contains redundant information.
In this case, it is useful to extract essential information from $X$ to describe the behavior 
of the quantum phenomena in the quantum system $Y$.
To discuss the essential information, we introduce the concept of 
$\epsilon$-(approximate) sufficient statistics
of the classical system $X$ with respect to the quantum system $Y$
while the papers \cite{PhysRevLett.117.090502,8115272} discussed 
this concept when system $Y$ is a classical system.

A function $f$ from $X$ to $T$ is called a sufficient statistics 
of $X$ for the quantum system $Y$
when there exists a conditional distribution $P_{X|T}$ such that
\begin{align}
\rho_{XY}= 
\sum_{t}P_{X|T}(x|t) |x\rangle \langle x| 
\otimes \sum_{x'\in f^{-1}(t)} P_X(x')  \rho_{Y|x'}.
\end{align}
The above condition is equivalent to the condition
\begin{align}
I(X:Y)=I(T:Y)
\end{align}
while in general we have the inequality $I(X:Y)\ge I(T:Y)$.

However, when we use sufficient statistics, we cannot remove 
a small correlation generated by a noise. 
As an example, suppose that the classical system $X$ is composed of 
two classical systems $X_1$ and $X_2$.  
Assume that we {have} a c-q state 
$\rho_{X_1X_2 Y}= 
\sum_{x_1} \sum_{x_2}P_{X_1,X_2}(x_1,x_2) |x_1,x_2\rangle \langle x_1,x_2| \otimes 
\rho_{Y|x_1}$ with 
two classical systems  $X_1$ and $X_2$. 

We assume that we have already {known} the distribution $P_{X_1 X_2}$
but we do not know $\rho_{Y|x}$.
Also, we assume that we generate this state several times and apply the state estimation 
to the generated state.
As {a} result, we obtain our estimate 
\begin{align}
\hat{\rho}_{X_1X_2Y}= 
\sum_{x_1} \sum_{x_2} P_{X_1X_2}(x_1,x_2) |x_1,x_2\rangle \langle x_1,x_2| \otimes 
\hat{\rho}_{Y|x_1,x_2}.\Label{AMO}
\end{align}
Since our estimate always has small error,
 $\hat{\rho}_{Y|x_1,x_2}$ is not exactly the same as $\rho_{Y|x_1}$, but it  is close to $\rho_{Y|x_1}$.
In this case, this difference should be considered as a noise.
That is, the dependence of $X_2$ is not essential.
It is better to consider that the correlation is given as
$\hat{\rho}_{Y|x_1}:=\sum_{x_2}P_{X_2|X_1}(x_2|x_1)\hat{\rho}_{Y|x_1,x_2}$ so that our estimate of $\rho_{X_1X_2Y}$ is given as
$\sum_{x_1} \sum_{x_2}P_{X_1,X_2}(x_1,x_2) |x_1,x_2\rangle \langle x_1,x_2| \otimes 
\hat{\rho}_{Y|x_1}$.
  
For $\epsilon>0$, 
a function $f:X\to T$ is called an $\epsilon$-sufficient statistics 
when the inequality
\begin{align}
I(X:Y) -\epsilon \le I(T:Y)
\end{align}
holds.
Hence, a sufficient statistics with $T$ of small size 
and an $\epsilon$-sufficient statistics 
can be considered as compressed data of $X$ with respect to $Y$.

In the above example, 
$X_1 X_2$ is a sufficient statistics for $Y$.
When $\delta$ is sufficiently small for $\epsilon$,
$I(X_1:Y) $ is close to $I(X_1X_2:Y) $, i.e., 
 $X_1$ is an $\epsilon$-sufficient statistics.
Hence, we can remove non-essential information $X_2$.
In fact, {if ${\cal X}={\cal X}_1\times {\cal X}_2$ is disturbed by a random permutation $\pi$,
it will be non-trivial to extract  essential information.}
To cover such a non-trivial case, we need a {systematic} approach 
to find such a function with a small-size $T$.
For this aim, we can use the information bottleneck algorithm.

To extract approximate sufficient statistics $T$, 
we focus on two requirements.
The mutual information $I(T:Y)$ should be larger,
and the entropy $H(T)$ should be smaller.
To handle these requirements, 
we simply minimize $H(T)-\beta I(T:Y)$ by using 
deterministic information bottleneck algorithm with $|{\cal T}|=|{\cal X}|$.
Since the algorithm minimizes $H(T)-\beta I(T:Y)$,
and the conditional distribution $P_{T|X}$ in the solution is deterministic,
the support of $P_T$ in the solution is expected to be smaller than
the original set ${\cal T}$.

\subsection{Numerics}

\begin{figure}[tb]{}
\begin{center} 
\includegraphics[width=0.6\linewidth]{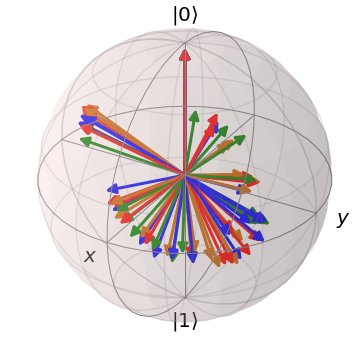}
\end{center}
\caption{{\bf Bloch representation of the estimated ensemble $\{\rho(\theta_{x_1,x_2},\lambda_{x_1,x_2})\}$.} As can be seen in the figure, the qubit states, especially those with higher purity, form several clusters in the Bloch ball. In each cluster, the states have the same value of $x_1$ and different values of $x_2$. This shows that the correlation between $X_1$ and $Y$ is higher than the correlation between $X_2$ and $Y$. }\Label{fig-bloch-DIB} 
\end{figure}

To demonstrate the above idea, let us take a look at a concrete example, which is a modification of the example in Section \ref{sec-var-gamma}.
Consider a single-qubit quantum system $Y$ and a classical register $X$ that encodes information about $Y$. 
The register $X$ is further split into two sub-registers $X_1$ and $X_2$ that take values in 
the sets ${\cal X}_1=\{0,1,\dots,4\}$ and 
${\cal X}_2=\{0,1,\dots,19\}$.  
Then, we assume that
$P_X$ is the uniform distribution over ${\cal X}_1\times {\cal X}_2$, and
the density $\rho_{Y|x_1}$ is given as 
$\rho(\theta_{x_1},\lambda_{x_1})$ with \eqref{numerics-rhox}.
The parameters $\theta$ and $\lambda$ depend on $x_1$ as
\begin{align}
\theta_{x_1}&:=\pi\cdot\frac{x_1}{|{\cal X}_1|}\qquad
\lambda_{x_1}:=\frac{x_1}{4|{\cal X}_1|}.
\end{align}

Obviously, the quantum system depends only on $X_1$ and $X_2$ contains no information about the quantum system. An experimentalist who has access to the ensemble, however, does not know this.
To extract information about the quantum system, for each pair of $(x_1,x_2)$, the experimentalist estimates its density matrix by repetitively (for $\nu<\infty$ times) making a suitable measurement on $\rho\left(\theta_{x_1},\lambda_{x_1}\right)$. 
According to quantum state estimation theory \cite{holevo2011probabilistic,helstrom1969quantum}, the estimate has an inaccuracy proportional to $1/\sqrt{\nu}$. Taking this into account, we model the estimated density matrix as $\rho\left(\theta_{x_1,x_2},\lambda_{x_1,x_2}\right)$ 
when the actual density matrix is $\rho\left(\theta_{x_1},\lambda_{x_1}\right)$, where
\begin{align}\Label{estimated-theta-lambda1}
\theta_{x_1,x_2}&:=\pi\cdot\frac{x_1}{|{\cal X}_1|}\left(1+r_\nu(x_1,x_2)\right)\\
\lambda_{x_1,x_2}&:=\frac{x_1}{4|{\cal X}_1|}\left(1+r'_\nu(x_1,x_2)\right)\Label{estimated-theta-lambda2}
\end{align}
and $r_\nu(x_1,x_2),r'_\nu(x_1,x_2)=O(1/\sqrt{\nu})$ characterise the estimation errors.
The estimated ensemble then admits the density matrix given in
\eqref{AMO} with $\hat{\rho}_{Y|x_1,x_2}=
\rho\left(\theta_{x_1,x_2},\lambda_{x_1,x_2}\right)$, which is
given by Eqs.~\eqref{numerics-rhox}, \eqref{estimated-theta-lambda1},
and \eqref{estimated-theta-lambda2}.
Notice that now the register $X_2$ is correlated with $Y$ in the estimated joint state $\hat{\rho}_{XY}$, even if the estimation-induced noise follows a distribution that does not depend on the value of $X_2$.

Now, the task is to compress the register $X$, by constructing a map from $X$ to a smaller {classical} register $T$. Here we take $T$ to be the same size as $X${.} 
One intuitive approach is to discard the $X_2$ register 
because $X_1$ contains much more information about the qubit state than $X_2$. 
Nevertheless, such a simple map does not exist in more general cases. 
For instance, if the values of $(x_1,x_2)$ in 
Eq.~(\ref{AMO}) are permuted, discarding $X_2$ will not result in faithful compression.
To see this, we further apply a arbitrary chosen unknown reshuffling $\pi:{\cal X}\to {\cal X}$ 
to the classical register ${\cal X}={\cal X}_1\times {\cal X}_2$ 
in Eq.~(\ref{AMO}). The ensemble then admits the following joint density matrix:
\begin{align}
\hat{\rho}'_{XY}
=&\sum_{x_1,x_2}\Big(P_{X}(x_1,x_2)|\pi(x_1,x_2)\>\<\pi(x_1,x_2)|\nonumber \\
&\otimes \rho\left(\theta_{x_1,x_2},\lambda_{x_1,x_2}\right)\Big)
\Label{DIB-state}
\end{align} 
with $\rho\left(\theta_{x_1,x_2},\lambda_{x_1,x_2}\right)$ given by Eqs.~(\ref{estimated-theta-lambda1}) and (\ref{estimated-theta-lambda2}). The goal is to extract an approximate sufficient statistics by constructing a map $Q:{\cal X}\to {\cal T}$.

\begin{widetext}
\begin{figure*}[bth]
\begin{center} 
\includegraphics[width=0.45\linewidth]{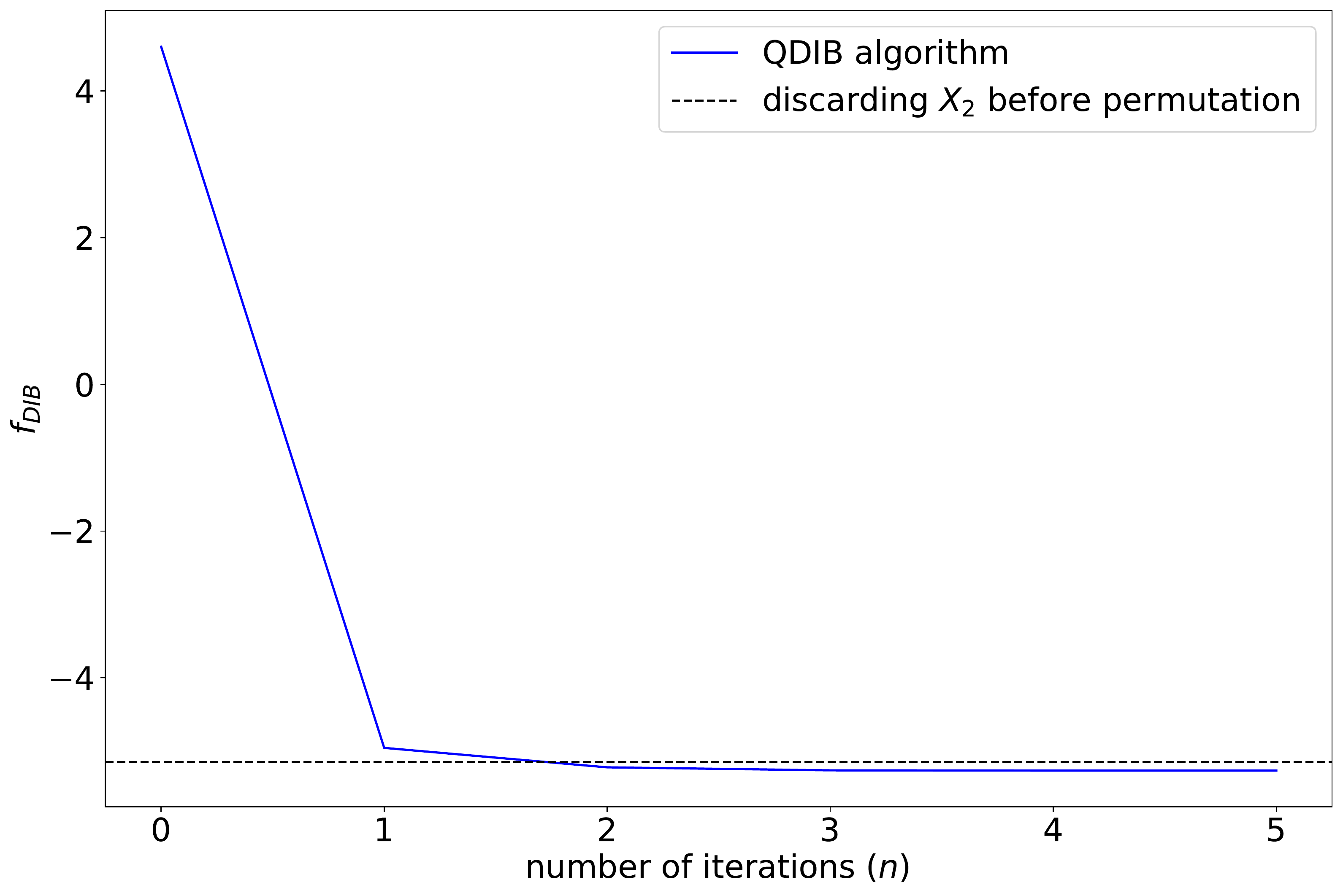}\qquad
\includegraphics[width=0.45\linewidth]{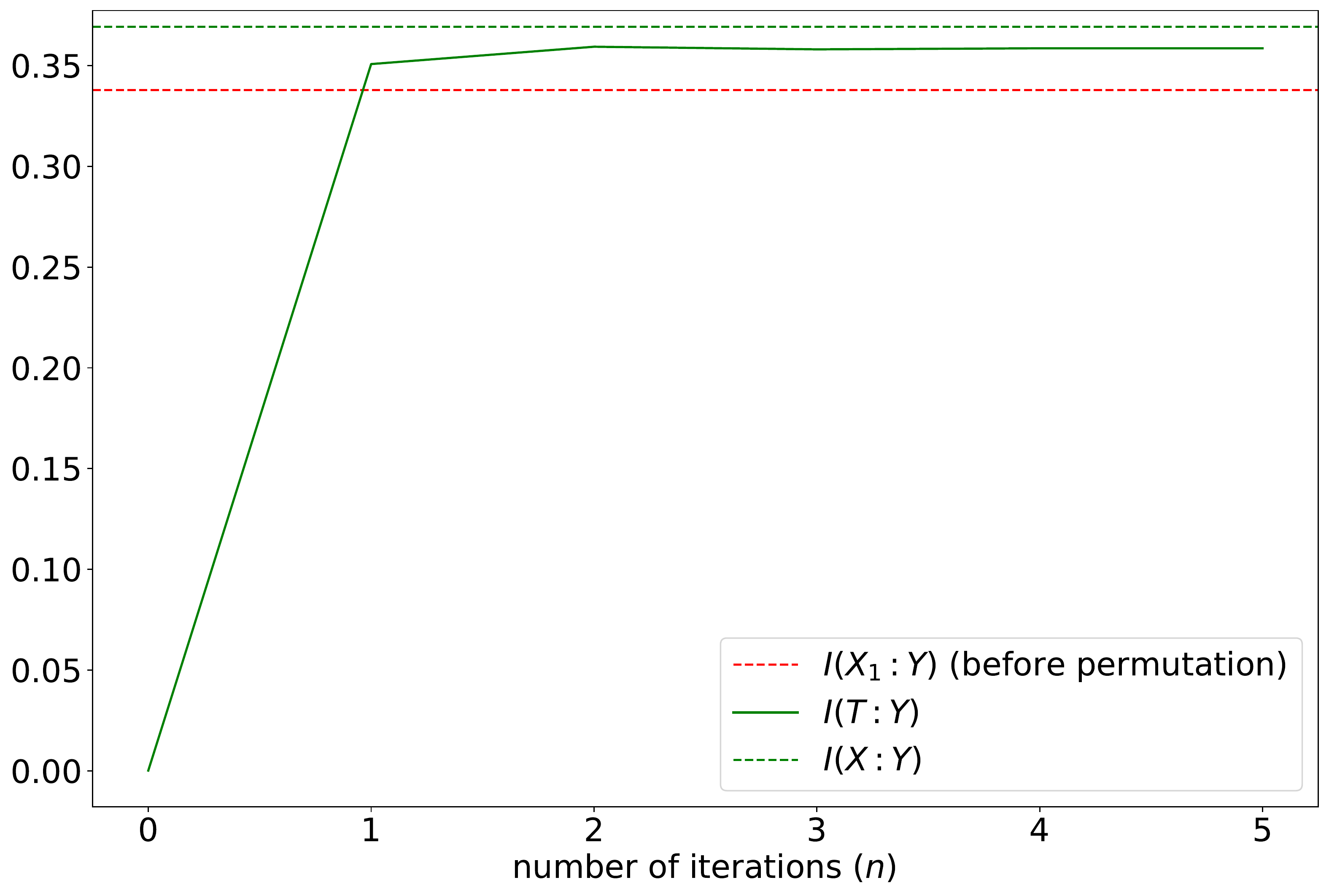}
\end{center}
\caption{{\bf Performance of QDIB algorithm in constructing approximate sufficient statistics.} We apply our quantum deterministic information bottleneck (QDIB) algorithm on the state (\ref{DIB-state}) (see also Fig.~\ref{fig-bloch-DIB}). For the joint state, we choose $|{\cal X}_1|=5$ 
and $|{\cal X}_2|=20$, and $P_X$ to be the uniform distribution over ${\cal X}={\cal X}_1\times {\cal X}_2$. The noise $r_\nu(x_1,x_2)$ and $r'_\nu(x_1,x_2)$ are drawn randomly and uniformly from the interval $(-1/\sqrt{\nu},1/\sqrt{\nu})$ with $\nu=20$ for any 
$x_1\in {\cal X}_1,x_2\in {\cal X}_2$.  In the QDIB algorithm 
(Algorithm \ref{protocol-DIB}), we choose $\beta=20$ and $|{\cal T}|=|{\cal X}|=100$.
In the figure above, the information bottleneck $f_{\rm DIB}:=f_{\alpha\to0}$ is plotted as a function of the number of iterations. 
As can be seen from the plot, the QDIB value of our algorithm becomes lower than that of the fictional protocol of ``discarding $X_2$ after the inverse permutation $\pi^{-1}$"   after only 3 iterations.
In the figure below, the faithfulness $I(T:Y)$ is plotted as a function of the number of iterations, and $I(X_1:Y)$ after the inverse permutation $\pi^{-1}$ 
(corresponding to the performance of the fictional protocol of 
``discarding $X_2$ after the inverse permutation $\pi^{-1}$") 
as well as $I(X:Y)$ (corresponding to the upper bound of $I(T:Y)$) are plotted for reference. 
Both plots justify that our QDIB algorithm performs well in the task of constructing approximate sufficient statistics.
}\Label{fig-DIB} 
\end{figure*}
\end{widetext}

Our QDIB algorithm works as a more systematic and more efficient method to extract essential information and discard non-essential information, 
even in the presence of an arbitrary permutation.  
In the QDIB algorithm (Algorithm \ref{protocol-DIB}), 
we choose $\beta=20$ and $|{\cal T}|=|{\cal X}|=|{\cal X}_1||{\cal X}_2|$.
First, we consider the case when the ensemble admits the form 
(\ref{AMO}), and the performance is summarised in Fig.~\ref{fig-DIB}. 
As one can see from the numerics, 
$f_{\rm DIB}:=f_{\alpha\to 0}$ of applying our QDIB algorithm to $\hat{\rho}_{XY}$ drops lower 
than that of the ``discarding $X_2$ after the inverse permutation $\pi^{-1}$'' 
approach within 5 iterations, 
and converges to a much lower value, 
suggesting a better compression performance. 
This is further justified in the second plot, where the faithfulness $I(T:Y)$ and the residual information $I(T:X)$ are plotted. 
We can see that 
since our QDIB algorithm preserves almost as much information about $Y$ 
as the original variable $X$,
it compresses a considerably larger portion of information about the original register $X$.

\section{Discussion and conclusion} \Label{S9} 
We have proposed a generalized algorithm for QIB with an acceleration parameter $\gamma$
and an additional parameter $\alpha$,
and have derived a {necessary} condition for the monotonic decrease of the objective function
$ f_{\alpha}=H(T)-\alpha H(T|X)-\beta I(T:Y) $
with quantum systems $Y,T$ and classical system $X$
when we extract information $T$ with respect to $Y$ from $X$. 
We have also showed its convergence under the same condition and that 
a wisely-chosen  parameter $\gamma$ can accelerate the convergence.
Our numerical calculation has further justified the above analysis as follows.
In our numerical experiment, making $\gamma$ smaller accelerates the convergence, but if $\gamma$ is made smaller than a threshold the algorithm will fail to converge.
In addition, we have provided examples that 
quantum system $T$ have an advantage over classical system $T$ 
even when $Y$ and $X$ {are classical}.

Next, taking the limit $\alpha\to +0$, 
we have proposed an iterative algorithm for QDIB that minimizes 
the objective function $ f_{\rm DIB}=H(T)-\beta I(T:Y) $.
We have shown that this iterative algorithm always decreases the objective function monotonically.
QDIB can be applied to find an approximate sufficient statistics because 
it realizes a smaller entropy $H(T)$ and a larger mutual information $I(T:Y) $.
Then, we have numerically demonstrated that our QDIB algorithm works well as
an approximate sufficient statistics.

An important application we show in this work is that our QIB algorithm yields a new approach of constructing quantum feature maps for classification.
In our numerical example, 
quantum system $T$ realizes a smaller value of the objective function 
than classical system $T$.
This numerical analysis shows the advantage of using quantum memory $T$ for the classification.
Despite significant recent progress \cite{wittek2014quantum,schuld2015introduction,biamonte2017quantum,schuld2019quantum,havlivcek2019supervised,blank2020quantum,lloyd2020quantum,perez2020data,schuld2021supervised}, the advantage of quantum machine learning over its classical counterpart has not been discussed much. Our work provides a new angle of attacking this issue, shedding light on a new proposal to rigorously justify and quantify quantum supremacy in the world of learning.  

 An open question left for future study is how to extend our result to the case where $X$ is also a quantum system, which covers, for instance,the scenario of compressing a quantum system while keeping its correlation with a classical label \cite{plesch2010efficient,rozema2014quantum,yang2016efficient,PhysRevLett.117.090502,yang2018compression,yang2018quantum}. Remarkably, in such a scenario, it has been shown that, if $T$ is classical, some correlation will be lost regardless of its size \cite{yang2018compression}. Therefore, we anticipate that the advantage of a quantum $T$ might persist or grow even stronger for QIB with a quantum $X$.

{
Finally, we remark that currently there is no efficient method to compute the restriction on $\gamma$ in Theorem \ref{TH3}. Resolving this important issue in a future work will accelerate the convergence of our information bottleneck algorithm.
}

\section*{Acknowledgement}
MH is supported in part by the National Natural Science Foundation of China (Grant No.~62171212) and Guangdong Provincial Key Laboratory (Grant No.~2019B121203002). YY is supported by Guangdong Basic and Applied Basic Research Foundation (Project No. 2022A1515010340), and by the Hong Kong Research Grant Council (RGC) through the Early Career Scheme (ECS) grant 27310822.

\bibliography{ref}

\begin{thebibliography}{38}
\providecommand{\natexlab}[1]{#1}
\providecommand{\url}[1]{\texttt{#1}}
\expandafter\ifx\csname urlstyle\endcsname\relax
  \providecommand{\doi}[1]{doi: #1}\else
  \providecommand{\doi}{doi: \begingroup \urlstyle{rm}\Url}\fi

\bibitem[Arimoto(1972)]{1054753}
S.~Arimoto.
\newblock An algorithm for computing the capacity of arbitrary discrete
  memoryless channels.
\newblock \emph{IEEE Transactions on Information Theory}, 18\penalty0
  (1):\penalty0 14--20, 1972.
\newblock \doi{10.1109/TIT.1972.1054753}.

\bibitem[Banchi et~al.(2021)Banchi, Pereira, and
  Pirandola]{PRXQuantum.2.040321}
Leonardo Banchi, Jason Pereira, and Stefano Pirandola.
\newblock Generalization in quantum machine learning: A quantum information
  standpoint.
\newblock \emph{PRX Quantum}, 2:\penalty0 040321, Nov 2021.
\newblock \doi{10.1103/PRXQuantum.2.040321}.

\bibitem[Biamonte et~al.(2017)Biamonte, Wittek, Pancotti, Rebentrost, Wiebe,
  and Lloyd]{biamonte2017quantum}
Jacob Biamonte, Peter Wittek, Nicola Pancotti, Patrick Rebentrost, Nathan
  Wiebe, and Seth Lloyd.
\newblock Quantum machine learning.
\newblock \emph{Nature}, 549\penalty0 (7671):\penalty0 195--202, 2017.
\newblock \doi{10.1038/nature23474}.

\bibitem[Blahut(1972)]{1054855}
R.~Blahut.
\newblock Computation of channel capacity and rate-distortion functions.
\newblock \emph{IEEE Transactions on Information Theory}, 18\penalty0
  (4):\penalty0 460--473, 1972.
\newblock \doi{10.1109/TIT.1972.1054855}.

\bibitem[Blank et~al.(2020)Blank, Park, Rhee, and
  Petruccione]{blank2020quantum}
Carsten Blank, Daniel~K Park, June-Koo~Kevin Rhee, and Francesco Petruccione.
\newblock Quantum classifier with tailored quantum kernel.
\newblock \emph{npj Quantum Information}, 6\penalty0 (1):\penalty0 1--7, 2020.
\newblock \doi{10.1038/s41534-020-0272-6}.

\bibitem[Datta et~al.(2019)Datta, Hirche, and Winter]{8849518}
Nilanjana Datta, Christoph Hirche, and Andreas Winter.
\newblock Convexity and operational interpretation of the quantum information
  bottleneck function.
\newblock In \emph{2019 IEEE International Symposium on Information Theory
  (ISIT)}, pages 1157--1161, 2019.
\newblock \doi{10.1109/ISIT.2019.8849518}.

\bibitem[Gily{\'e}n et~al.(2019)Gily{\'e}n, Su, Low, and
  Wiebe]{gilyen2019quantum}
Andr{\'a}s Gily{\'e}n, Yuan Su, Guang~Hao Low, and Nathan Wiebe.
\newblock Quantum singular value transformation and beyond: exponential
  improvements for quantum matrix arithmetics.
\newblock In \emph{Proceedings of the 51st Annual ACM SIGACT Symposium on
  Theory of Computing}, pages 193--204, 2019.
\newblock \doi{10.1145/3313276.3316366}.

\bibitem[Goldfeld and Polyanskiy(2020)]{goldfeld2020information}
Ziv Goldfeld and Yury Polyanskiy.
\newblock The information bottleneck problem and its applications in machine
  learning.
\newblock \emph{IEEE Journal on Selected Areas in Information Theory},
  1\penalty0 (1):\penalty0 19--38, 2020.
\newblock \doi{10.1109/JSAIT.2020.2991561}.

\bibitem[Grimsmo and Still(2016)]{PhysRevA.94.012338}
Arne~L. Grimsmo and Susanne Still.
\newblock Quantum predictive filtering.
\newblock \emph{Phys. Rev. A}, 94:\penalty0 012338, Jul 2016.
\newblock \doi{10.1103/PhysRevA.94.012338}.

\bibitem[Harrow et~al.(2009)Harrow, Hassidim, and Lloyd]{harrow2009quantum}
Aram~W Harrow, Avinatan Hassidim, and Seth Lloyd.
\newblock Quantum algorithm for linear systems of equations.
\newblock \emph{Physical review letters}, 103\penalty0 (15):\penalty0 150502,
  2009.
\newblock \doi{10.1103/PhysRevLett.103.150502}.

\bibitem[Havl{\'\i}{\v{c}}ek et~al.(2019)Havl{\'\i}{\v{c}}ek, C{\'o}rcoles,
  Temme, Harrow, Kandala, Chow, and Gambetta]{havlivcek2019supervised}
Vojt{\v{e}}ch Havl{\'\i}{\v{c}}ek, Antonio~D C{\'o}rcoles, Kristan Temme,
  Aram~W Harrow, Abhinav Kandala, Jerry~M Chow, and Jay~M Gambetta.
\newblock Supervised learning with quantum-enhanced feature spaces.
\newblock \emph{Nature}, 567\penalty0 (7747):\penalty0 209--212, 2019.
\newblock \doi{10.1038/s41586-019-0980-2}.

\bibitem[Hayashi and Tan(2018)]{8115272}
Masahito Hayashi and Vincent Y.~F. Tan.
\newblock Minimum rates of approximate sufficient statistics.
\newblock \emph{IEEE Transactions on Information Theory}, 64\penalty0
  (2):\penalty0 875--888, 2018.
\newblock \doi{10.1109/TIT.2017.2775612}.

\bibitem[Helstrom(1969)]{helstrom1969quantum}
Carl~W Helstrom.
\newblock Quantum detection and estimation theory.
\newblock \emph{Journal of Statistical Physics}, 1\penalty0 (2):\penalty0
  231--252, 1969.
\newblock \doi{10.1007/BF01007479}.

\bibitem[Hirche and Winter(2020)]{9174416}
Christoph Hirche and Andreas Winter.
\newblock An alphabet-size bound for the information bottleneck function.
\newblock In \emph{2020 IEEE International Symposium on Information Theory
  (ISIT)}, pages 2383--2388, 2020.
\newblock \doi{10.1109/ISIT44484.2020.9174416}.

\bibitem[Holevo(2011)]{holevo2011probabilistic}
Alexander~S Holevo.
\newblock \emph{Probabilistic and statistical aspects of quantum theory},
  volume~1.
\newblock Springer Science \& Business Media, 2011.
\newblock \doi{10.1007/978-88-7642-378-9}.

\bibitem[Hsu et~al.(2006)Hsu, Kennedy, and Chang]{10.1145/1180639.1180654}
Winston~H. Hsu, Lyndon~S. Kennedy, and Shih-Fu Chang.
\newblock Video search reranking via information bottleneck principle.
\newblock MM '06, pages 35--44, New York, NY, USA, 2006. Association for
  Computing Machinery.
\newblock ISBN 1595934472.
\newblock \doi{10.1145/1180639.1180654}.

\bibitem[Lloyd et~al.(2020)Lloyd, Schuld, Ijaz, Izaac, and
  Killoran]{lloyd2020quantum}
Seth Lloyd, Maria Schuld, Aroosa Ijaz, Josh Izaac, and Nathan Killoran.
\newblock Quantum embeddings for machine learning.
\newblock \emph{arXiv preprint arXiv:2001.03622}, 2020.
\newblock \doi{10.48550/arXiv.2001.03622}.

\bibitem[Low and Chuang(2017)]{low2017hamiltonian}
Guang~Hao Low and Isaac~L Chuang.
\newblock Hamiltonian simulation by uniform spectral amplification.
\newblock \emph{arXiv preprint arXiv:1707.05391}, 2017.
\newblock \doi{10.48550/arXiv.1707.05391}.

\bibitem[Low and Chuang(2019)]{low2019hamiltonian}
Guang~Hao Low and Isaac~L Chuang.
\newblock Hamiltonian simulation by qubitization.
\newblock \emph{Quantum}, 3:\penalty0 163, 2019.
\newblock \doi{10.22331/q-2019-07-12-163}.

\bibitem[P{\'e}rez-Salinas et~al.(2020)P{\'e}rez-Salinas, Cervera-Lierta,
  Gil-Fuster, and Latorre]{perez2020data}
Adri{\'a}n P{\'e}rez-Salinas, Alba Cervera-Lierta, Elies Gil-Fuster, and
  Jos{\'e}~I Latorre.
\newblock Data re-uploading for a universal quantum classifier.
\newblock \emph{Quantum}, 4:\penalty0 226, 2020.
\newblock \doi{10.22331/q-2020-02-06-226}.

\bibitem[Plesch and Bu{\v{z}}ek(2010)]{plesch2010efficient}
Martin Plesch and Vladim{\'\i}r Bu{\v{z}}ek.
\newblock Efficient compression of quantum information.
\newblock \emph{Physical Review A}, 81\penalty0 (3):\penalty0 032317, 2010.
\newblock \doi{10.1103/PhysRevA.81.032317}.

\bibitem[Ramakrishnan et~al.(2021)Ramakrishnan, Iten, Scholz, and
  Berta]{Ramakrishnan}
Navneeth Ramakrishnan, Raban Iten, Volkher~B. Scholz, and Mario Berta.
\newblock Computing quantum channel capacities.
\newblock \emph{IEEE Transactions on Information Theory}, 67\penalty0
  (2):\penalty0 946--960, 2021.
\newblock \doi{10.1109/TIT.2020.3034471}.

\bibitem[Rozema et~al.(2014)Rozema, Mahler, Hayat, Turner, and
  Steinberg]{rozema2014quantum}
Lee~A Rozema, Dylan~H Mahler, Alex Hayat, Peter~S Turner, and Aephraim~M
  Steinberg.
\newblock Quantum data compression of a qubit ensemble.
\newblock \emph{Physical Review Letters}, 113\penalty0 (16):\penalty0 160504,
  2014.
\newblock \doi{10.1103/PhysRevLett.113.160504}.

\bibitem[Salek et~al.(2019)Salek, Cadamuro, Kammerlander, and Wiesner]{8513885}
Sina Salek, Daniela Cadamuro, Philipp Kammerlander, and Karoline Wiesner.
\newblock Quantum rate-distortion coding of relevant information.
\newblock \emph{IEEE Transactions on Information Theory}, 65\penalty0
  (4):\penalty0 2603--2613, 2019.
\newblock \doi{10.1109/TIT.2018.2878412}.

\bibitem[Schuld(2021)]{schuld2021supervised}
Maria Schuld.
\newblock Supervised quantum machine learning models are kernel methods.
\newblock \emph{arXiv preprint arXiv:2101.11020}, 2021.
\newblock \doi{10.48550/arXiv.2101.11020}.

\bibitem[Schuld and Killoran(2019)]{schuld2019quantum}
Maria Schuld and Nathan Killoran.
\newblock Quantum machine learning in feature {H}ilbert spaces.
\newblock \emph{Physical Review Letters}, 122\penalty0 (4):\penalty0 040504,
  2019.
\newblock \doi{10.1103/PhysRevLett.122.040504}.

\bibitem[Schuld et~al.(2015)Schuld, Sinayskiy, and
  Petruccione]{schuld2015introduction}
Maria Schuld, Ilya Sinayskiy, and Francesco Petruccione.
\newblock An introduction to quantum machine learning.
\newblock \emph{Contemporary Physics}, 56\penalty0 (2):\penalty0 172--185,
  2015.
\newblock \doi{10.1080/00107514.2014.964942}.

\bibitem[Shwartz-Ziv and Tishby(2017)]{shwartz2017opening}
Ravid Shwartz-Ziv and Naftali Tishby.
\newblock Opening the black box of deep neural networks via information.
\newblock \emph{arXiv preprint arXiv:1703.00810}, 2017.
\newblock \doi{10.48550/arXiv.1703.00810}.

\bibitem[Slonim and Tishby(2000)]{10.1145/345508.345578}
Noam Slonim and Naftali Tishby.
\newblock Document clustering using word clusters via the information
  bottleneck method.
\newblock SIGIR '00, pages 208--215, New York, NY, USA, 2000. Association for
  Computing Machinery.
\newblock ISBN 1581132263.
\newblock \doi{10.1145/345508.345578}.

\bibitem[Stark et~al.(2018)Stark, Shah, and Bauch]{8368978}
Maximilian Stark, Aizaz Shah, and Gerhard Bauch.
\newblock Polar code construction using the information bottleneck method.
\newblock In \emph{2018 IEEE Wireless Communications and Networking Conference
  Workshops (WCNCW)}, pages 7--12, 2018.
\newblock \doi{10.1109/WCNCW.2018.8368978}.

\bibitem[Strouse and Schwab(2017)]{10.1162/NECO_a_00961}
DJ~Strouse and David~J. Schwab.
\newblock {The Deterministic Information Bottleneck}.
\newblock \emph{Neural Computation}, 29\penalty0 (6):\penalty0 1611--1630, 06
  2017.
\newblock ISSN 0899-7667.
\newblock \doi{10.1162/NECO_a_00961}.

\bibitem[Tishby et~al.(1999)Tishby, Pereira, and Bialek]{tishby1999}
N.~Tishby, F.~C. Pereira, and W.~Bialek.
\newblock The information bottleneck method.
\newblock In \emph{The 37th annual Allerton Conference on Communication,
  Control, and Computing}, pages 368--377. Univ. Illinois Press, 1999.
\newblock \doi{10.48550/arXiv.physics/0004057}.

\bibitem[Tishby and Zaslavsky(2015)]{tishby2015deep}
Naftali Tishby and Noga Zaslavsky.
\newblock Deep learning and the information bottleneck principle.
\newblock In \emph{2015 IEEE information theory workshop (ITW)}, pages 1--5.
  IEEE, 2015.
\newblock \doi{10.1109/ITW.2015.7133169}.

\bibitem[Wittek(2014)]{wittek2014quantum}
Peter Wittek.
\newblock \emph{Quantum machine learning: what quantum computing means to data
  mining}.
\newblock Academic Press, 2014.
\newblock \doi{10.1016/C2013-0-19170-2}.

\bibitem[Yang et~al.(2016{\natexlab{a}})Yang, Chiribella, and
  Ebler]{yang2016efficient}
Yuxiang Yang, Giulio Chiribella, and Daniel Ebler.
\newblock Efficient quantum compression for ensembles of identically prepared
  mixed states.
\newblock \emph{Physical Review Letters}, 116\penalty0 (8):\penalty0 080501,
  2016{\natexlab{a}}.
\newblock \doi{10.1103/PhysRevLett.116.080501}.

\bibitem[Yang et~al.(2016{\natexlab{b}})Yang, Chiribella, and
  Hayashi]{PhysRevLett.117.090502}
Yuxiang Yang, Giulio Chiribella, and Masahito Hayashi.
\newblock Optimal compression for identically prepared qubit states.
\newblock \emph{Phys. Rev. Lett.}, 117:\penalty0 090502, Aug
  2016{\natexlab{b}}.
\newblock \doi{10.1103/PhysRevLett.117.090502}.

\bibitem[Yang et~al.(2018{\natexlab{a}})Yang, Bai, Chiribella, and
  Hayashi]{yang2018compression}
Yuxiang Yang, Ge~Bai, Giulio Chiribella, and Masahito Hayashi.
\newblock Compression for quantum population coding.
\newblock \emph{IEEE Transactions on Information Theory}, 64\penalty0
  (7):\penalty0 4766--4783, 2018{\natexlab{a}}.
\newblock \doi{10.1109/TIT.2017.2788407}.

\bibitem[Yang et~al.(2018{\natexlab{b}})Yang, Chiribella, and
  Hayashi]{yang2018quantum}
Yuxiang Yang, Giulio Chiribella, and Masahito Hayashi.
\newblock Quantum stopwatch: how to store time in a quantum memory.
\newblock \emph{Proceedings of the Royal Society A: Mathematical, Physical and
  Engineering Sciences}, 474\penalty0 (2213):\penalty0 20170773,
  2018{\natexlab{b}}.
\newblock \doi{10.1098/rspa.2017.0773}.

\end{thebibliography}
\bibliographystyle{plainnat}

\end{document}